\pdfoutput=1

\documentclass[11pt]{article}

\usepackage{amsmath}
\usepackage{amsfonts}
\usepackage{amssymb}
\usepackage{latexsym}
\usepackage{graphicx}
\usepackage{psfrag}
\usepackage{stmaryrd}
\usepackage[all,cmtip]{xy}

\newcommand{\opn}[1]{\operatorname{#1}}

\newcommand{\real}{\opn{Re}}
\newcommand{\imag}{\opn{Im}}

\newcommand{\der}[2]{\frac{\partial #1}{\partial #2}}

\newcommand{\pr}{\partial}

\newcommand{\mc}[1]{\mathcal{#1}}
\newcommand{\jt}{\textstyle}
\newcommand{\jd}{\displaystyle}
\newcommand{\js}{\scriptstyle}

\newcommand{\wtil}[1]{\widetilde{#1}}

\newcommand{\IGNORE}[1]{}

\usepackage{epsfig}
\newcommand{\mypsdraft}{\psdraft}
\renewcommand{\mypsdraft}{\psfull}
\newcommand{\mypsfull}{\psfull}

\usepackage{theorem}
\theoremstyle{plain}
\theoremheaderfont{\scshape}
\newtheorem{theorem}{Theorem}

\title{An Infinite Branching Hierarchy of Time-Periodic Solutions of
the Benjamin-Ono Equation}

\date{November 3, 2008}

\author{
Jon Wilkening
\thanks{Department of Mathematics and Lawrence Berkeley National
  Laboratory, University of California, Berkeley, CA 94720 ({\tt
    wilken@math.berkeley.edu}).  This work was supported in part by the
  Director, Office of Science, Computational and Technology Research,
  U.S. Department of Energy under Contract No. DE-AC02-05CH11231.}
}

\begin{document}

\maketitle

\begin{abstract}
  We present a new representation of solutions of the Benjamin-Ono
  equation that are periodic in space and time.  Up to an additive
  constant and a Galilean transformation, each of these solutions is a
  previously known, multi-periodic solution; however, the new
  representation unifies the subset of such solutions with a fixed
  spatial period and a continuously varying temporal period into a
  single network of smooth manifolds connected together by an infinite
  hierarchy of bifurcations.  Our representation explicitly describes
  the evolution of the Fourier modes of the solution as well as the
  particle trajectories in a meromorphic representation of these
  solutions; therefore, we have also solved the problem of finding
  periodic solutions of the ordinary differential equation governing
  these particles, including a description of a bifurcation mechanism
  for adding or removing particles without destroying periodicity.  We
  illustrate the types of bifurcation that occur with several
  examples, including degenerate bifurcations not predicted by
  linearization about traveling waves.
\end{abstract}

\vspace*{5pt}

\vspace*{5pt}

\thispagestyle{empty}

\section{Introduction}
\label{sec:intro}

The Benjamin-Ono equation is a model water wave equation for the
propagation of unidirectional, weakly nonlinear internal waves in a
deep, stratified fluid \cite{benjamin:67, davis:67, ono:75}.  It is a
non-linear, non-local dispersive equation that, after a suitable
choice of spatial and temporal scales, may be written
\begin{equation} \label{eqn:BO:intro}
  u_t = Hu_{xx} - uu_x, \qquad
  Hf(x) = \frac{1}{\pi}\,PV\!\!\int_{-\infty}^\infty
  \frac{f(\xi)}{x-\xi}\,d\xi.
\end{equation}
Our motivation for studying time-periodic solutions of this equation
was inspired by the analysis of Plotnikov, Toland and Iooss
\cite{tolandPlotnikov, tolandPlotnikovIooss} using the Nash-Moser
implicit function theorem to prove the existence of non-trivial time
periodic solutions of the two-dimensional water wave over an
irrotational, incompressible, inviscid fluid.  We hope to learn more
about these solutions through direct numerical simulation.  As a first
step, in collaboration with D. Ambrose, the author has developed a
numerical continuation method \cite{benj1, benj2}
for the computation of time-periodic
solutions of non-linear PDE and used it to compute
families of time-periodic solutions of the Benjamin-Ono equation,
which shares many of the features of the water wave such as
non-locality, but is much less expensive to compute.

Because we came to this problem from the perspective of developing
numerical tools that can also be used to study the full water wave
equation, we did not take advantage of the existence of solitons or
complete integrability in our numerical study of the Benjamin-Ono
equation.  The purpose of the current paper is to bridge this
connection, i.e.~to show how the form of the exact solutions we
deduced from numerical simulations is related to previously known,
multi-periodic solutions \cite{satsuma:ishimori:79, dobro:91,
  matsuno:04}.  Our representation is quite different, describing 
time-periodic solutions in terms of the trajectories of the Fourier
modes, which are expressed in terms of $N$ particles $\beta_j(t)$
moving through the unit disk of the complex plane.  Thus, one of the
main results of this paper is to show the relationship between the
meromorphic solutions described e.g.~in~\cite{case:BO:N:sol} and these
multi-periodic solutions.  Our representation also simplifies the
computation and visualization of multi-periodic solutions.  Rather
than solve a system of non-linear algebraic equations at each $x$ to
find $u(x)$ as was done in \cite{matsuno:04}, we represent $u(x)$
through its Fourier coefficients by finding the zeros $\beta_j$ of a
polynomial whose coefficients involve only a finite number of non-zero
temporal Fourier modes.  We find that plotting the trajectories of the
particles $\beta_j(t)$ often gives more information about the solution
than making movies of $u(x,t)$ directly.

A key difference in our setup of the problem is that we wish to fix
the spatial period once and for all (using e.g.~$2\pi$) and describe
all families of time-periodic solutions in which the temporal period
depends continuously on the parameters of the family.  This framework
may be perceived as awkward and overly restrictive by some readers,
and we agree that the most natural ``periodic''
generalization of the $N$-soliton solutions \cite{matsuno:BO:sol,
  matsuno:BO:04} of an integrable system such as Benjamin-Ono are the
$N$-phase quasi-periodic (or multi-periodic) solutions
\cite{hino:almost:periodic, satsuma:ishimori:79, dobro:91,
  matsuno:04}.  However, our goal in this paper is not to study the
behavior of these solutions in the long wave-length limit, but rather
to understand how all these families of solutions are connected
(continuously) together through a hierarchy of bifurcations.  The
additional restriction of exact periodicity makes the bifurcation
problem harder for integrable problems, but easier for other problems
that can only be studied numerically.  To our knowledge, bifurcation
between levels in the hierarchy of multi-periodic solutions of
Benjamin-Ono has not previously been discussed. Indeed, with the
exception of \cite{dobro:91}, previous representations of these
solutions are missing a key degree of freedom, the mean, which must
vary in order for these solutions to connect with each other.  The
most interesting result of this paper is that counting dimensions of
nullspaces of the linearized problem does not predict certain
degenerate bifurcations that allow for immediate jumps across several
levels of the infinite hierarchy of time-periodic solutions.
As a consequence, in our numerical studies, we
found bifurcations from traveling waves to the second level of the
hierarchy, and interior bifurcations from these solutions to the third
level of the hierarchy, but never saw bifurcations from traveling
waves directly to the third level of the hierarchy (as we did not know
where to look for them).  This will be important to keep in mind in
problems such as the water wave, where exact solutions are not
expected to be found.

We believe we have accounted for all time-periodic solutions of the
Benjamin-Ono equation with a fixed spatial period, but do not know how
to prove this.  Even for the simplest case of a traveling wave, it is
surprisingly difficult to prove that the solitary and periodic waves
found by Benjamin \cite{benjamin:67} are the only possibilities;
see \cite{amick:toland:BO} and Appendix~\ref{appendix:unique}.  For the
closely related KdV equation \cite{ablowitz:segur, novikov, whitham},
substitution of $u(x,t)=u_0(x-ct)$ into the
equation leads to an ordinary differential equation for $u_0(x)$ with
periodic solutions involving Jacobi elliptic functions; see
e.g.~\cite{osborne:95}.  However, for Benjamin-Ono, the equation for
the traveling wave shape
is non-local due to the Hilbert transform.  Nevertheless, Amick and
Toland \cite{amick:toland:BO} have shown that any traveling wave
solution of Benjamin-Ono can be extended to the upper half-plane to
agree with the real part of a bounded, holomorphic function satisfying
a complex ODE; thus, in spite of non-locality, we are able to obtain
uniqueness by solving an initial value problem for $u_0(x)$.
Interestingly, these traveling wave shapes are rational functions of
$e^{ix}$, which are simpler than the cnoidal solutions of KdV.  This
analysis of traveling waves via holomorphic extension to the upper
half-plane is similar in spirit to the analysis of rapidly decreasing
solutions of Benjamin-Ono via the inverse scattering transform
\cite{fokas:ablowitz:83, kaup:matsuno:98}.  Thus, it may be possible
to prove that we have accounted for all periodic solutions of
Benjamin-Ono by developing a spatially periodic version of the IST
that is analogous to the study of Bloch eigenfunctions and Riemann
surfaces for the periodic KdV equation \cite{novikov, belokolos}, but
this has not yet been carried out.

This paper is organized as follows.  In Section~\ref{sec:mero}, we
describe meromorphic solutions \cite{case:BO:N:sol} of the
Benjamin-Ono equation, which are a class of solutions represented by a
system of $N$ particles evolving in the upper half-plane (or, in our
representation, the unit disk) according to a completely integrable
ODE.  We also show the relationship between the elementary symmetric
functions of these particles and the Fourier coefficients of the
solution, which were observed in numerical experiments to have
trajectories in the complex plane consisting of epicycles involving a
finite number of circular orbits.  In Section~\ref{sec:trav}, we
summarize the results of \cite{benj2} in the form of a theorem (not
proved in \cite{benj2}) enumerating all bifurcations from traveling
waves to the second level of the hierarchy of time-periodic solutions.
The new idea that allows us to prove the theorem is to show that all
the zeros of a certain polynomial lie inside the unit circle using
Rouch\'e's theorem.  In Section~\ref{sec:hierarchy}, we state a
theorem that parametrizes solutions at level $M$ of the hierarchy
through an explicit description of the particle trajectories
$\beta_j(t)$.  This theorem also describes the way in which different
levels of the hierarchy are connected together through bifurcation.
The proof of this theorem shows the relationship with previous studies
of multi-periodic solutions.  In Section~\ref{sec:examples}, we give
several examples of degenerate and non-degenerate bifurcations between
various levels of the hierarchy.  We also use these examples to
illustrate some of the topological changes that occur in the particle
trajectories along the paths of solutions between bifurcation states.
Finally, in Appendix~\ref{appendix:roots}, we give a direct proof that
the particles $\beta_j(t)$ in our formulas lie inside the unit disk in
the complex plane; in Appendix~\ref{appendix:unique}, we discuss
uniqueness of traveling wave solutions.

\section{Meromorphic Solutions and Particle Trajectories}
\label{sec:mero}

In this section,
we consider spatially periodic solutions of the Benjamin-Ono equation,
\begin{equation} \label{eqn:BO}
  u_{t}=Hu_{xx}-uu_{x}.
\end{equation}
Here $H$ is the Hilbert transform defined in (\ref{eqn:BO:intro}),
which has the symbol $\hat{H}(k)=-i\operatorname{sgn}(k).$ It is well
known \cite{case:BO:N:sol} that meromorphic solutions of
(\ref{eqn:BO}) of the form
\begin{equation} \label{eqn:mero1}
  u(x,t) = 2\real\left\{\sum_{l=1}^N \frac{2k}{e^{ik[x+kt-x_l(t)]}-1}\right\}
\end{equation}
exist, where $k$ is a real wave number and the particles $x_l(t)$
evolve in the upper half of the complex plane according to the
equation
\begin{equation} \label{eqn:xl:ode}
  \frac{dx_l}{dt} =
  \sum_{\parbox{.3in}{$\js m=1\\[-8pt]m\ne l$}}^N
  \frac{2k}{e^{-ik(x_m-x_l)}-1} +
  \sum_{m=1}^N\frac{2k}{e^{-ik(x_l-\bar{x}_m)}-1}, \qquad (1\le l\le N).
\end{equation}
This representation is useful for studying the dynamics of solitons of
(\ref{eqn:BO}) over $\mathbb{R}$, which may be obtained from
(\ref{eqn:mero1}) in the long wave-length limit $k\rightarrow0$.
However, over a fixed periodic domain $\mathbb{R}\big/2\pi\mathbb{Z}$,
we have found it more convenient to work with particles $\beta_l(t)$
evolving in the unit disk of the complex plane,
\begin{equation}
  \beta_l = e^{-i\bar{x}_l} \in \Delta:=\{z\,:\,|z|<1\},
\end{equation}
where a bar denotes complex conjugation.  Up to an additive constant
and a Galilean transformation, the solution $u$ in (\ref{eqn:mero1})
may then be written
\begin{equation} \label{eqn:u:rep}
  u(x,t) = \alpha_0 + \sum_{l=1}^N u_{\beta_l(t)}(x),
\end{equation}
where we have included the mean $\alpha_0$ as an additional parameter
of the solution and $u_\beta(x)$ is defined via
\begin{equation} \label{eqn:u:beta0}
  u_\beta(x) = \frac{4|\beta|\{\cos(x-\theta)-|\beta|\}}{
    1+|\beta|^2 - 2|\beta|\cos(x-\theta)}, \qquad
  \big(\beta = |\beta|e^{-i\theta}\big).
\end{equation}
From (\ref{eqn:xl:ode}) or direct substitution into (\ref{eqn:BO}),
the $\beta_l$ are readily shown to satisfy
\begin{equation} \label{eqn:beta:ode}
  \dot{\beta}_l = \sum_{\parbox{.3in}{$\js m=1\\[-8pt]m\ne l$}}^N
  \frac{2i}{\beta_l^{-1}-\beta_m^{-1}} +
  \sum_{m=1}^N \frac{2i\beta_l^2}{\beta_l - \bar{\beta}_m^{-1}}
  + i(1-\alpha_0)\beta_l, \qquad (1\le l\le N).
\end{equation}
It is awkward to work with $u_\beta(x)$ in physical space; however, in
Fourier space, it takes the simple form
\begin{equation} \label{u:beta:hat}
  \hat{u}_{\beta,k} =
  \begin{cases}
    2\bar{\beta}^{|k|}, & \quad k<0, \\
    0, & \quad k=0, \\ 2\beta^k, & \quad k>0.
  \end{cases}
\end{equation}
As a result, the Fourier coefficients $c_k(t)$
of $u(x,t)$ in (\ref{eqn:u:rep}) are simply power sums of the particle
trajectories,
\begin{equation}
  c_k(t) = \begin{cases} \alpha_0 & \quad k=0, \\
    2\bigl[\beta_1^k(t) + \cdots + \beta_N^k(t)\bigr], & \quad k>0,
    \end{cases}
\end{equation}
where $c_k=\bar{c}_{-k}$ for $k<0$.  In \cite{benj2}, it was found
numerically that although the $\beta_l$ often execute very complicated
periodic orbits, the elementary symmetric functions $\sigma_j$
defined via
\begin{equation}
  \sigma_0 = 1, \qquad \sigma_j = \sum_{l_1<\cdots<l_j}\beta_{l_1}\cdots
  \beta_{l_j}, \qquad (j=1,\dots,N)
\end{equation}
have orbits that are circles (or at worst, epicycles involving a
finite number of non-zero temporal Fourier coefficients) in the
complex plane.  As a consequence, the spatial Fourier coefficients
$c_k=2\opn{tr}(\Sigma^k)$, where $k\ge1$ and $\Sigma$ is the companion
matrix of the polynomial $P(z)=\prod(z-\beta_j)$, also have
trajectories that are epicycles, which we noticed immediately in our
numerical simulations.  All the solutions in this paper will be of the
form
\begin{equation}
  \prod_{l=1}^N [z-\beta_l(t)] =
  \sum_{j=0}^N (-1)^j\sigma_j(t)z^{N-j} = P(z,e^{-i\omega t}),
\end{equation}
where
\begin{equation} \label{eqn:P:laurent}
  P(z,\lambda)=\sum_{j=0}^N (-1)^j\tilde\sigma_j(\lambda)z^{N-j}
\end{equation}
is a monic polynomial in $z$ with coefficients
$\tilde\sigma_j$ that are Laurent polynomials in $\lambda$,
and such that for any $\lambda$ on the unit circle $S^1$ in the
complex plane, the roots $\beta_1$, \dots,$\beta_N$ of $P(\cdot,\lambda)$
lie inside the unit disk.

We may express the solution $u$ in (\ref{eqn:u:rep}) directly
in terms of $P$ as follows:
\begin{align}\label{eqn:u:from:P}
  u(x,t) &= \alpha_0 + \sum_{l=1}^N u_{\beta_l(t)}(x) =
  \alpha_0 + \sum_{l=1}^N 4\real\left\{\sum_{k=1}^\infty
    \beta_l(t)^k e^{ikx}\right\}  \\
  \notag
  &= \alpha_0 + \sum_{l=1}^N 4\real\left\{\frac{z}{z-\beta_l(t)}-1\right\}
  = \alpha_0 + 4\real\left\{\frac{z\pr_zP(z,\lambda)}{P(z,\lambda)}-N\right\},
  \quad \left(\parbox{.7in}{\centering$z = e^{-ix}$\\
    $\lambda=e^{-i\omega t}$}\right).
\end{align}
Note that $Hu=4\real\left\{-i\left[(z\partial z P)/P -
    N\right]\right\}$.  The choice $z=e^{-ix}$ (as opposed to
$e^{+ix}$) follows from the decision in (\ref{u:beta:hat}) to have
Fourier coefficients with positive wave numbers carry powers of
$\beta$ rather than $\bar{\beta}$, while the choice
$\lambda=e^{-i\omega t}$ leads to a natural sign convention when we
interpret the exponents in the Laurent polynomials
$\tilde{\sigma}_j(\lambda)$ in (\ref{eqn:P:laurent}) as measures of
the direction and velocity of the traveling waves obtained in certain
limits.
It was shown in \cite{benj2} that (\ref{eqn:u:from:P}) is a solution
of (\ref{eqn:BO}) if there is a constant $\gamma\in\mathbb{R}$ such
that
\begin{equation} \label{eqn:P:alg}
  \begin{aligned}
  \gamma P_{00}\bar{P}_{00}
  & + \bar{P}_{00}\big[P_{20} + \omega P_{01} +
  (\alpha_0-4N)P_{10}\big] \\
  & + P_{00}\big[\bar{P}_{20} + \omega\bar{P}_{01} +
  (\alpha_0-4N)\bar{P}_{10}\big] +
  2P_{10}\bar{P}_{10}=0,
  \end{aligned}
\end{equation}
where
\begin{equation}
  P_{jk}=(z\partial_z)^j(\lambda\partial_\lambda)^kP(z,\lambda)\biggr
  \vert_{\parbox[b]{.45in}{$\js z=e^{-ix}$ \\[-9pt]
      $\js \lambda=e^{-i\omega t}$}}.
\end{equation}
The goal of this paper is to find explicit formulas for the solutions
$P(z,\lambda)$ of (\ref{eqn:P:alg}), show how they fit in with the
previously known families of multi-periodic solutions described in
\cite{satsuma:ishimori:79, dobro:91, matsuno:04}, and determine how
these families are connected together through bifurcation.

\section{Paths connecting arbitrary traveling waves}
\label{sec:trav}

In \cite{benj2}, a classification of bifurcations from traveling waves
was proposed after all time-periodic solutions of the linearization of
(\ref{eqn:BO}) about traveling waves were found in closed form.  A
numerical continuation method was then developed to follow paths of
non-trivial time-periodic solutions beyond the realm of validity of
the linearization until another traveling wave was reached (or until
the solution blows up as the bifurcation parameter approaches a
critical value).  Through extensive data fitting of the numerical
solutions, the exact form of the solutions on this path was deduced.
In this section, we give an alternative formula for these exact
solutions that unifies the three cases described in \cite{benj2} and
makes it possible to show that the roots $\beta_l$ of the polynomial
$P(\cdot,\lambda)$ are inside the unit circle for $\lambda\in S^1$.

An $N$-hump traveling wave is uniquely determined by the mean,
$\alpha_0$, a complex number $\beta\in\Delta$, and a positive integer,
$N$:
\begin{equation} \label{eqn:N:trav}
  u_{\alpha_0,N,\beta}(x,t) =
  \alpha_0 + \sum_{l=1}^N u_{\beta_l(t)}(x), \qquad
  \beta_l(t) = \sqrt[N]\beta e^{-ict}, \qquad
  c = \alpha_0 - N\frac{1-3|\beta|^2}{1-|\beta|^2}.
\end{equation}
Here $\beta_l$ ranges over all $N$th roots of $\beta$.  This
solution may also be written
\begin{equation}
  u_{\alpha_0,N,\beta}(x,t) = u_{N,\beta}(x-ct)+c, \qquad
  u_{N,\beta} = N\frac{1-3|\beta|^2}{1-|\beta|^2} + Nu_\beta(Nx),
\end{equation}
where $u_{N,\beta}(x)=Nu_{1,\beta}(Nx)$ is the $N$-hump stationary
solution; hence, the traveling wave moves to the right if $c>0$.  We
can solve for $c$ and $\alpha_0$ in terms of the period
$T>0$ and a speed index $\nu\in\mathbb{Z}$ indicating how many
increments of $\frac{2\pi}{N}$ the wave moves to the right in one
period:
\begin{equation}
  cT = \frac{2\pi\nu}{N}, \qquad
  \alpha_0 = c + N\frac{1-3|\beta|^2}{1-|\beta|^2}.
\end{equation}
In order to bifurcate to a non-trivial time-periodic
solution, the period $T$ must be related to an eigenvalue
\begin{equation}
  \omega_{N,n} = \begin{cases} (n)(N-n), & \quad 1\le n\le N-1, \\
    (n+1-N)\Big[n+1+N\Big(1-\frac{1-3|\beta|^2}{1-|\beta|^2}\Big)\Big],
    & \quad n\ge N
  \end{cases}
\end{equation}
of the linear operator \cite{benj2} governing the evolution of
solutions of the linearization about the $N$-hump stationary solution:
\begin{equation}\label{eqn:omega:T}
  \omega_{N,n}T = \frac{2\pi m}{N}, \qquad 1\le m \in \begin{cases}
    n\nu + N\mathbb{Z}, & \quad 1\le n < N, \\
    (n+1)\nu + N\mathbb{Z}, & \quad n\ge N.
  \end{cases}
\end{equation}
This requirement
on the oscillation index $m$
enforces the condition that the linearized solution over the
stationary solution return to a phase shift of itself to account for
the fact that the traveling wave has moved during this time; see
\cite{benj2}.  Here we have used the fact that if
$u(x,t)=u_{N,\beta}(x)$ is a stationary solution and
\begin{equation}\label{eqn:add:const}
  U(x,t)=u(x-ct,t)+c
\end{equation}
is a traveling wave, then the solutions $v$ and $V$ of the
linearizations about $u$ and $U$, respectively, satisfy
$V(x,t)=v(x-ct,t)$.  The parameter $\beta$ together
with the four integers $(N,\nu,n,m)$ enumerate the bifurcations from
traveling waves, which comprise the first level of the hierarchy
of time-periodic solutions of the Benjamin-Ono equation, to the second
level of this hierarchy.  We will see later that other
bifurcations from traveling waves to higher levels of the hierarchy
also exist, which is interesting as they are not predicted from
counting dimensions of nullspaces in the linearization.

After (numerically) mapping out which bifurcations $(N,\nu,n,m)$ and
$(N',\nu',n',m')$ were connected by paths of non-trivial solutions,
it was found that $N$, $N'$, $\nu$ and $\nu'$ can be chosen
independently as long as
\begin{equation}
  N'<N, \qquad\qquad \nu' > \frac{N'}{N}\nu.
\end{equation}
The other parameters are then given by
\begin{equation}
  m = m' = N\nu'-N'\nu>0, \qquad n = N-N', \qquad n' = N-1.
\end{equation}
The following theorem proves that these numerical conjectures
are correct.

\begin{theorem} \label{thm1}
  Let $N$, $N'$, $\nu$ and $\nu'$ be integers satisfying
  $N>N'>0$ and $m=N\nu'-N'\nu>0$.  There is a four-parameter family of
  time-periodic solutions connecting the traveling wave bifurcations
  $(N',\nu',N-1,m)$ and $(N,\nu,N-N',m)$.  These solutions are of the
  form
\begin{equation}
  u(x,t) = \alpha_0 + \sum_{l=1}^N u_{\beta_l(t-t_0)}(x-x_0),
\end{equation}
where $\beta_1(t)$, \dots, $\beta_N(t)$ are the roots of the
polynomial $P(\cdot,e^{-i\omega t})$ defined by
\begin{equation}\label{eqn:P:def3}
  P(z,\lambda) = z^N + A\lambda^{\nu'}z^{N-N'} +
  B\lambda^{\nu-\nu'}z^{N'} + C\lambda^{\nu},
\end{equation}
with
\begin{gather}
  \label{eqn:ABC}
  A = \sqrt{\frac{N-N'+s+s'}{N+s+s'}}
  \sqrt{\frac{(N+s')s'}{N'(N-N')+(N+s')s'}}, \\[2pt]
\notag
  B = \sqrt{\frac{(N+s')s'}{N'(N-N')+(N+s')s'}}
  \sqrt{\frac{s}{N-N'+s}}, \qquad 
  C = \sqrt{\frac{s}{N-N'+s}}
  \sqrt{\frac{N-N'+s+s'}{N+s+s'}}, \\[2pt]
\notag
  \alpha_0 = \frac{N^2\nu'-(N')^2\nu}{m} - 2s - \frac{2N'(\nu'-\nu)}{m}s',
  \qquad \omega = \frac{2\pi}{T} = \frac{N'(N-N')(N+2s')}{m}.
\end{gather}
The four parameters are $s\ge0$, $s'\ge0$,
$x_0\in\mathbb{R}$ and $t_0\in\mathbb{R}$.  The $N$- and $N'$-hump
traveling waves occur when $s'=0$ and $s=0$, respectively.  When
both are zero, we obtain the constant solution $u(x,t)\equiv
\frac{N^2\nu'-(N')^2\nu}{m}$.
\label{thm:4}
\end{theorem}

\begin{proof}
  Without loss of generality, we may assume $x_0=0$ and $t_0=0$.
  It was shown in \cite{benj2} that $P(z,\lambda)$ of the
  form (\ref{eqn:P:def3}) satisfies (\ref{eqn:P:alg}) if
\begin{align}
  &\gamma = (3N-\alpha_0)N - \nu\omega, \\
  \label{eqn:alg1}
&  [(N')^2 - 2 N N' + N'\alpha_0 - \nu' \omega]B + 
  [(N')^2 + 2 N N' - N'\alpha_0 + \nu' \omega]AC = 0, \\
   \notag
& \bigl[3N^2 - 4NN' + (N')^2 - (N-N')\alpha_0 + (\nu-\nu')\omega\bigr]BC
\\
& \hspace*{2in}
   \label{eqn:alg2}
   - \bigl[N^2 - (N')^2 - (N-N')\alpha_0 + (\nu-\nu')\omega\bigr]A = 0,
\\ \notag
& (N\alpha_0 - \nu\omega - N^2) +
\big[(2N'-N)\alpha_0 + (\nu-2\nu')\omega + 3N^2 -8NN' + 4(N')^2\big]B^2 \\
\label{eqn:alg3}
&\quad + \big[(N-2N')\alpha_0 + 4(N')^2-N^2 +
(2\nu'-\nu)\omega\big]A^2 + \big[(3N-\alpha_0)N+\nu\omega\big]C^2 = 0.
\end{align}
Using a computer algebra system, it is easy to check that
(\ref{eqn:alg1})--(\ref{eqn:alg3}) hold when $A$, $B$, $C$, $\alpha_0$
and $\omega$ are defined as in (\ref{eqn:ABC}).  When $s'=0$, we have
$A=B=0$ and $\jd C=\sqrt{\frac{s}{N+s}}$\, so that
\begin{equation*}
  \beta_l(t) = \sqrt[N]{-C\lambda^{\nu}} = \sqrt[N]{-C}e^{-ict}, \qquad
  c = \frac{\omega\nu}{N} = \frac{N'(N-N')\nu}{m} = \alpha_0 - N
  \frac{1-3C^2}{1-C^2},
\end{equation*}
where each $\beta_l$ is assigned a distinct $N$th root of $-C$.  By
(\ref{eqn:N:trav}), this is an $N$-hump traveling wave with speed index
$\nu$ and period $T=\frac{2\pi}{\omega}$.  Similarly, when $s=0$, we
have $B=C=0$ and $\jd A=\sqrt{\frac{s'}{N'+s'}}$\, so that
\begin{equation*}
  \beta_l(t)=\left\{\begin{array}{cc}
      \sqrt[N']{-A}e^{-ict} & l\le N' \\
      0 & l>N' \end{array}\right\}, \qquad
  c = \frac{\omega\nu'}{N'} = \frac{(N-N')(N+2s')\nu'}{m} =
  \alpha_0 - N'\frac{1-3A^2}{1-A^2},
\end{equation*}
which is an $N'$-hump traveling wave with speed index $\nu'$ and
period $T=\frac{2\pi}{\omega}$.

Finally, we show that the roots of $P(\cdot,\lambda)$ are
inside the unit disk for any $\lambda$ on the unit circle, $S^1$.
We will use Rouch\'e's theorem \cite{ahlfors}.  Let
\begin{alignat*}{4}
    f_1(z) &= z^N & + \,&A\lambda^{\nu'}z^{N-N'} & &+B\lambda^{\nu-\nu'}z^{N'}
    & &+ C\lambda^\nu, \\
    f_2(z) &= z^N & + \,&A\lambda^{\nu'}z^{N-N'},\hspace*{-.1in} \\
    f_3(z) &= z^N & & & &+B\lambda^{\nu-\nu'}z^{N'}.\hspace*{-.1in}
\end{alignat*}
From (\ref{eqn:ABC}), we see that $\{A,B,C\}\subseteq[0,1)$,
$A\ge BC$, $B\ge CA$ and $C\ge AB$.  Thus,
\begin{equation}
\begin{aligned}
  d_2(z) := |f_2(z)|^2 &- |f_1(z)-f_2(z)|^2 = |\lambda^{-\nu'}z^{N'}+A|^2 -
  |B\lambda^{-\nu'}z^{N'}+C|^2 \\
  &= 1 + A^2 - B^2 - C^2 + 2(A-BC)\cos \theta \ge
  (1-A)^2 - (B-C)^2,
\end{aligned}
\end{equation}
where $\lambda^{-\nu'}z^{N'}=e^{i\theta}$.  Similarly,
\begin{equation}
  d_3 := |f_3(z)|^2 - |f_1(z)-f_3(z)|^2 \ge (1-B)^2 - (A-C)^2.
\end{equation}
Note that
\begin{alignat*}{3}
  &B\le A, \quad C\le B \quad & &\Rightarrow \quad
  B-C \le B-AB < 1-A \quad & &\Rightarrow \quad
  d_2(z)>0 \text{ for } z\in S^1, \\
  &B\le A, \quad C>B \quad & &\Rightarrow \quad
  |C-A|<1-B \quad & &\Rightarrow \quad
  d_3(z)>0 \text{ for } z\in S^1, \\
  &A\le B, \quad C\le A \quad & &\Rightarrow \quad
  A-C \le A-AB < 1-B \quad & &\Rightarrow \quad
  d_3(z)>0 \text{ for } z\in S^1, \\
  &A\le B, \quad C>A \quad & &\Rightarrow \quad
  |C-B|<1-A \quad & &\Rightarrow \quad
  d_2(z)>0 \text{ for } z\in S^1.
\end{alignat*}
Thus, in all cases, $f_1(z)=P(z,\lambda)$ has the same number of zeros
inside $S^1$ as $f_2(z)$ or $f_3(z)$, which each have $N$ roots inside
$S^1$.  Since $f_1(z)$ is a polynomial of degree $N$, all the roots are
inside $S^1$.
\end{proof}

\section{An infinite hierarchy of interior bifurcations}
\label{sec:hierarchy}

Next we wish to find all possible cascades of interior bifurcations
from these already non-trivial solutions to more and more complicated
time-periodic solutions.
The most interesting consequence of the following theorem is that
there are some traveling waves with more bifurcations to non-trivial
time-periodic solutions than predicted by counting the dimension of
the kernel of the linearization of the map measuring deviation from
time-periodicity.  This will be illustrated in various examples in
Section~\ref{sec:examples}.

\begin{theorem} \label{thm2}
Let $M\ge2$ be an integer and choose
\begin{align*}
  &k_1,\dots,k_M\in\mathbb{N},
  \qquad \text{(positive integers, not necessarily distinct or monotonic),} \\
  &\nu_1,\dots,\nu_M\in\mathbb{Z}, \qquad \text{(arbitrary integers
    satisfying $\nu_{j-1}<\frac{k_{j-1}}{k_j}\nu_j$ for $j\ge2$)}.
\end{align*}
Now define the positive quantities
\begin{align}
  \label{eqn:mj:def}
  &m_j = k_{j-1}\nu_j - k_j\nu_{j-1}, \quad
   \tau_j = \frac{k_j(k_j + k_{j-1})k_{j-1}}{m_j}, \quad
  \gamma_j = \frac{2k_jk_{j-1}}{m_j}, \quad (2\le j\le M).
\end{align}
Let $J=\opn{argmax}_{2\le j\le M} \tau_j$.  If there is a tie, $J$ can
be any of the candidates.  Then there is an $M+2$ parameter family of
time-periodic solutions of the Benjamin-Ono equation parametrized by
\begin{equation}
  s_1\ge0, \qquad s_J\ge0, \qquad x_{j0}\in\mathbb{R}, \qquad
  (1\le j\le M)
\end{equation}
and constructed as follows.  First, we define
\begin{gather} \label{eqn:sj1}
  s_j = \frac{\tau_J-\tau_j}{\gamma_j} + \frac{\gamma_J}{\gamma_j}s_J,
  \qquad\quad (2\le j\le M, \;\; j\ne J), \\
  q_1 = s_1, \qquad p_1 = s_1 + k_1, \qquad
  q_j = p_{j-1}+s_j, \qquad p_j = q_j + k_j, \qquad (2\le j\le M)
\end{gather}
so that $s_j\ge0$ for $1\le j\le M$ and
\begin{equation} \label{eqn:qp}
  \xymatrix@C-2pc@R-1.5pc{
    0 \ar@/^1pc/[rr]^{s_1}
    & \le & q_1 \ar@/^1pc/[rr]^{k_1}
    & < & p_1 \ar@/^1pc/[rr]^{s_2}
    & \le & q_2 \ar@/^1pc/[rr]^{k_2}
    & < & \cdots \ar@/^1pc/[rr]^{s_{M-1}}
    & \le & q_{M-1} \ar@/^1pc/[rr]^{k_{M-1}}
    & < & p_{M-1} \ar@/^1pc/[rr]^{s_M}
    & \le & q_M \ar@/^1pc/[rr]^{k_M}
    & < & p_M.
  }
\end{equation}
For any subset $S$ of $\mc{M}=\{1,\dots,M\}$, we denote
the complement by $S'=\mc{M}\setminus S$ and define
\begin{equation} \label{eqn:CS:def}
  k_S = \sum_{j\in S}k_j, \qquad
  \nu_{S'} = \sum_{m\in S'} \nu_m, \qquad
  C_S = \Biggl( \prod_{(j,m) \in S\times S'} a_{jm}\Biggr)
  \Biggl( \prod_{m\in S'} b_m \Biggr),
\end{equation}
where
\begin{equation} \label{eqn:ab:def}
  a_{jm} = \sqrt{\frac{(p_m-q_j)(q_m-p_j)}{(q_m-q_j)(p_m-p_j)}}, \qquad
  b_m = \sqrt{\frac{q_m}{p_m}}e^{-ik_mx_{m0}}.
\end{equation}
Then
\begin{equation} \label{eqn:u:a0:beta}
  u(x,t) = \alpha_0 + \sum_{l=1}^N u_{\beta_l(t)}(x)
\end{equation}
is a periodic solution of the Benjamin-Ono equation, where
$\jt N = k_\mc{M} = \sum_{j=1}^M k_j$,
\begin{equation} \label{eqn:a0:w}
  \alpha_0 = \biggl(2N-k_1+\frac{\nu_1}{k_1}\tau_J\biggr) 
  - 2s_1 + \frac{\nu_1}{k_1}\gamma_J s_J, \qquad
  \omega = \frac{2\pi}{T} = \tau_J + \gamma_J s_J,
\end{equation}
and $\beta_1(t),\dots,\beta_N(t)$ are the roots of the polynomial
$z\mapsto P(z,e^{-i\omega t})$ given by
\begin{equation} \label{eqn:P:CS}
  P(z,\lambda) = \sum_{S\in\mc{P}(\mc{M})}
  C_S \lambda^{\nu_{S'}}z^{k_S}.
\end{equation}
When $M=2$, this representation coincides with that of Theorem~\ref{thm1}
if we set
\begin{alignat*}{4}
  k_1 &= N-N', &\qquad \nu_1 &= \nu-\nu', &\qquad s_1 &= s, & \qquad
  x_{10} &= x_0 - \frac{\nu_1}{k_1}\omega t_0, \\
  k_2 &= N', &\qquad \nu_2 &= \nu', &\qquad s_2 &= s',
  & \qquad x_{20} &= x_0 - \frac{\nu_2}{k_2}\omega t_0.
\end{alignat*}
It reduces to a traveling wave when $s=0$ or $s'=0$ and to a constant
solution when both are zero.  Similarly, when
$M\ge3$, the solution reduces to a simpler solution in this same
hierarchy (with $M$ replaced by $\wtil{M}=M-1$) when $s_1$ or
$s_J$ reaches zero.  Specifically, if $s_1=0$, then
$P(z,\lambda) = z^{k_1}\wtil{P}(z,\lambda)$, where $\wtil{P}(z,\lambda)$
corresponds to the parameters
\begin{equation} \label{eqn:s1eq0}
  \tilde{k}_j = k_{j+1}, \quad
  \tilde{\nu}_j = \nu_{j+1}, \quad
  \tilde{s}_j = s_{j+1}, \quad 
  \tilde{x}_{j0} = x_{j+1,0}, \quad
  (1\le j\le \wtil{M}).
\end{equation}
We interpret this as an annihilation of $k_1$
particles $\beta_l$ at the origin.  If $s_J=0$,
we have $P(z,\lambda) = \wtil{P}(z,\lambda)$, where
\begin{equation}\label{eqn:sJeq0}
\begin{array}{lllll}
  \tilde{k}_j = k_j, &
  \tilde{\nu}_j = \nu_j, &
  \tilde{s}_j = s_j, &
  \tilde{x}_{j0} = x_{j0}, &
  (1\le j\le J-2), \\
  \tilde{k}_j = k_j+k_J, &
  \tilde{\nu}_j = \nu_j+\nu_J, &
  \tilde{s}_j = s_j, &
  \tilde{x}_{j0} = \frac{k_jx_{j0}+k_Jx_{J0}}{k_j+k_J}, &
  (j=J-1), \\
  \tilde{k}_j = k_{j+1}, &
  \tilde{\nu}_j = \nu_{j+1}, &
  \tilde{s}_j = s_{j+1}, &
  \tilde{x}_{j0} = x_{j+1,0},\;\; &
  (J\le j\le \wtil{M}).
\end{array}
\end{equation}
If several $s_j$ are zero when $s_J=0$ (i.e.~if a tie occurs when
choosing $J=\opn{argmax}_{2\le j\le M}\tau_j$), the bifurcation is
degenerate and any subset of the indices for which $s_j=0$ may be
removed using the rule (\ref{eqn:sJeq0}) repeatedly (once for each
index removed, with $J$ ranging over these indices in reverse order to
avoid re-labeling), allowing for bifurcations that jump several levels
in the hierarchy at once.
\end{theorem}

\begin{proof}
  Rather than show directly that $P(z,\lambda)$ in (\ref{eqn:P:CS})
  satisfies (\ref{eqn:P:alg}), we show that each of our solutions
  differs from a multi-periodic solution
  \cite{satsuma:ishimori:79,matsuno:04} by at most a transformation of
  the form (\ref{eqn:add:const}).  We give a direct proof that all the
  zeros of $P(\cdot,\lambda)$ lie inside the unit circle in
  Appendix~\ref{appendix:roots} as concluding this from the combined
  results of \cite{dobro:91} and \cite{matsuno:04} is complicated.
  If all the inequalities in (\ref{eqn:qp}) are strict, then it is
  known \cite{matsuno:04} that
\begin{equation}\label{eqn:u:satsu}
  U = 2i\der{}{x}\log\frac{f'}{f}
\end{equation}
satisfies the Benjamin-Ono equation (\ref{eqn:BO}) with
\begin{gather}
  \label{eqn:f:satsu1}
  f' = \sum_{\mu=0,1}\exp\left[
    \sum_{j=1}^M \mu_j\left(i\theta_j - \phi_j -
      \frac{1}{2}\sum_{m\ne j}^{M}A_{jm}\right) +
    \sum_{j<m}^{(M)} \mu_j\mu_m A_{jm}\right], \\[2pt]
  \label{eqn:f:satsu2}
  f = \sum_{\mu=0,1}\exp\left[
    \sum_{j=1}^M \mu_j\left(i\theta_j + \phi_j -
      \frac{1}{2}\sum_{m\ne j}^{M}A_{jm}\right) +
    \sum_{j<m}^{(M)} \mu_j\mu_m A_{jm}\right], \\
  \theta_j = k_j(x-c_jt-x_{j0}), \qquad
  e^{2\phi_j} = \frac{p_j}{q_j}, \qquad
  k_j = p_j - q_j, \qquad c_j = p_j + q_j, \\
  \exp(A_{jm}) = \frac{(q_m-q_j)(p_m-p_j)}{(p_m-q_j)(q_m-p_j)}
  = \frac{(c_m-c_j)^2-(k_m-k_j)^2}{(c_m-c_j)^2-(k_m+k_j)^2},
\end{gather}
where $x_{j0}\in\mathbb{R}$ are $M$ arbitrary phase parameters and the
notation $\sum_{\mu=0,1}$ indicates a summation over all possible
combinations of $\mu_1=0,1$; $\mu_2=0,1$; \dots; $\mu_M=0,1$.  (The
notation $\sum_{j<m}^{(M)}$ indicates that $j$ and $m$ both vary
between $1$ and $M$ such that $j<m$, while $\sum_{m\ne j}^M$ indicates
that $m$ varies from $1$ to $M$ omitting $m=j$).

We write $f'$ and $f$ in (\ref{eqn:f:satsu1}) and
(\ref{eqn:f:satsu2}) as sums over all subsets $S$ of
$\mc{M}=\{1,\dots,M\}$:
\begin{equation}
  f',f\, = \sum_{S\in\mc{P}(\mc{M})}
  \bigg(\prod_{(j,m)\in S\times S'} e^{-\frac{1}{2}A_{jm}}\bigg)
  \prod_{j\in S} e^{ik_j(x-c_jt-x_{j0})\mp\phi_j}, \qquad
  (S' = \mc{M}\setminus S),
\end{equation}
where $-\phi_j$ is used for $f'$ and $+\phi_j$ is used for $f$.
Next we observe that
\begin{equation}\label{eqn:u:from:U}
  u(x,t)=U(x-ct,t)+c
\end{equation}
will be time-periodic with period~$T=\frac{2\pi}{\omega}$ if there
exist integers $\nu_j\in\mathbb{Z}$ such that
\begin{equation} \label{eqn:kj:c}
  k_j(c+c_j) = \nu_j\omega, \qquad (1\le j\le M).
\end{equation}
Then we have
\begin{gather}
  f'(x-ct,t) =
  \bigg(\prod_{j\in\mc{M}}e^{ik_jx}
    \bigg) \sum_{S\in\mc{P}(\mc{M})} C_S
    \bigg(\prod_{m\in S'}e^{-i\nu_m\omega t}\bigg)\bigg(
    \prod_{j\in S}e^{-ik_jx}\bigg), \\
\notag
  f(x-ct,t) = 
  \bigg(\prod_{j\in\mc{M}}e^{-ik_jx_{j0}}e^{-i\nu_j\omega t}e^{\phi_j}
    \bigg) \sum_{S\in\mc{P}(\mc{M})} \overline{C_S}
    \bigg(\prod_{m\in S'}e^{i\nu_m\omega t}\bigg)\bigg(
    \prod_{j\in S}e^{ik_jx}\bigg),
\end{gather}
with $C_S$ as in (\ref{eqn:CS:def}) above.  The complex conjugation in
$f$ comes from interchanging $S$ and $S'$ in the sum after factoring
out $\big(\prod_{j\in\mc{M}}\cdots\big)$.  It follows that $u$ in
(\ref{eqn:u:from:U}) is given by
\begin{equation} \label{eqn:u:from:gh}
  u = c-2N + 2i\partial_x \log \frac{g}{h} =
  \alpha_0 +
  2\left(\frac{i\partial_xg}{g}-N\right) +
  2\left(\frac{-i\partial_xh}{h}-N\right),
\end{equation}
where
\begin{equation}
  \begin{gathered}
    N = \sum_{j=1}^M k_j, \qquad \alpha_0 = c+2N, \qquad
    g(x,t)=P(e^{-ix},e^{-i\omega t}), \qquad h=\bar{g}, \\
    P(z,\lambda) = \sum_{S\in\mc{P}(\mc{M})}C_S \lambda^{\nu_{S'}}z^{k_S},
    \qquad \nu_{S'} = \sum_{m\in S'}\nu_m, \qquad
    k_S = \sum_{j\in S}k_j.
  \end{gathered}
\end{equation}
With $\lambda$ fixed, $P$ is a monic polynomial in $z$ of degree $N$.
If we complexify $x$ and fix $t$, then
\begin{equation}
  P(e^{-ix},e^{-i\omega t})=0 \quad \Leftrightarrow \quad
  f'(x-ct,t)=0 \quad \Leftrightarrow \quad
  f(\bar{x}-ct,t)=0,
\end{equation}
so all the zeros of $P$ are inside the unit circle iff all the zeros
of $f$ are in the upper half-plane and all the zeros of $f'$ are in
the lower half-plane.  These properties of $f$ and $f'$ were assumed
to be true in \cite{satsuma:ishimori:79}, leaving a small gap in
their proof (acknowledged in the paper); we give a proof in
Appendix~\ref{appendix:roots}.  The right hand
side of (\ref{eqn:u:from:gh}) is equal to the right hand side of
(\ref{eqn:u:from:P}), which establishes the representation
(\ref{eqn:u:a0:beta}) of $u(x,t)$ in terms of the trajectories
$\beta_l(t)$ of the roots of $P(\cdot,e^{-i\omega t})$.

Eliminating $c$ from (\ref{eqn:kj:c}) and using
$(c_j-c_{j-1})=(k_{j-1}+k_j+2s_j)$, we find that
\begin{equation}
  k_jk_{j-1}(k_{j-1}+k_j+2s_j) = (k_{j-1}\nu_j - k_j\nu_{j-1})\omega,
  \qquad (2\le j\le M).
\end{equation}
This shows that $m_j$ in (\ref{eqn:mj:def}) must be positive.
Eliminating $\omega$, we find that
\begin{equation} \label{eqn:tjsj}
  \tau_j + \gamma_js_j = \tau_J + \gamma_J s_J, \quad\qquad
  j,J\in\{2,\dots,M\}.
\end{equation}
Choosing $J=\opn{argmax}_{2\le j\le M}\tau_j$ and solving
(\ref{eqn:tjsj}) for $s_j$ in terms of $s_J$ yields (\ref{eqn:sj1}),
which ensures that $s_j>0$ whenever $s_J>0$.  From
\begin{equation}
  \alpha_0 = c+2N, \qquad
  c = -c_1 + \frac{\nu_1}{k_1}\omega, \qquad c_1 = k_1 + 2s_1, \qquad
  \omega = \tau_J+\gamma_Js_J,
\end{equation}
we obtain the formulas in (\ref{eqn:a0:w}) for $\alpha_0$ and
$\omega$.

Finally, we drop the assumption that the inequalities in
(\ref{eqn:qp}) are strict and observe what happens to these solutions
when $s_1=0$ or $s_J=0$.  If $s_1=0$, then $b_1=0$, so
\begin{equation*}
  1\not\in S \; \Rightarrow \; C_S = 0, \qquad
  1\in S \; \Rightarrow \; C_S =
  \Biggl(\prod_{(j,m)\in (S\setminus\{1\})\times S'} a_{jm}
  \Biggr)\Biggl(\prod_{m\in S'} a_{1m}b_m\Biggr).
\end{equation*}
But since $q_1=0$ and $p_1=k_1$ when $s_1=0$, we have
\begin{equation}
  a_{1m}b_m = \sqrt{\frac{(p_m-q_1)(q_m-p_1)}{(q_m-q_1)(p_m-p_1)}}
  \sqrt{\frac{q_m}{p_m}}e^{-ik_mx_{m0}} =
  \sqrt{\frac{q_m - k_1}{p_m - k_1}}e^{-ik_mx_{m0}}.
\end{equation}
Thus, if we define $\wtil{M}=M-1$ and shift indices down as in
(\ref{eqn:s1eq0}), the parameters $q_j$ and $p_j$ will decrease
by $k_1$ as illustrated here,
\begin{equation}
  \xymatrix@-1.7pc{
    0 \ar@/^1pc/[rr]^{s_1=0}
    & = & q_1 \ar@/^1pc/[rr]^{k_1}
    & < & p_1 \ar@/^1pc/[rr]^{s_2}
    & < & q_2 \ar@/^1pc/[rr]^{k_2}
    & < & p_2 \ar@/^1pc/[rr]^{s_3}
    & \le & q_3 \ar@/^1pc/[rr]^{k_3}
    & < & \cdots \ar@/^1pc/[rr]^{s_M}
    & \le & q_M \ar@/^1pc/[rr]^{k_M}
    & < & p_M \\
    & & & &
    0 \ar@/_1pc/[rr]_{\tilde{s}_1 = s_2}
    & \le & \tilde{q}_1 \ar@/_1pc/[rr]_{\tilde{k}_1 = k_2}
    & < & \tilde{p}_1 \ar@/_1pc/[rr]_{\tilde{s}_2 = s_3}
    & \le & \tilde{q}_2 \ar@/_1pc/[rr]_{\tilde{k}_2 = k_3}
    & < & \cdots \ar@/_1pc/[rr]_{\tilde{s}_{M-1} = s_M}
    & \le & \tilde{q}_{M-1} \ar@/_1pc/[rr]_{\tilde{k}_{M-1}=k_M}
    & < & \tilde{p}_{M-1}
  }
\end{equation}
and hence $\tilde{a}_{jm}=a_{j+1,m+1}$ and $\tilde{b}_{m}=
a_{1,m+1}b_{m+1}$ for $1\le j,m\le\wtil{M}$, $j\ne m$.  Therefore,
with the notation $\wtil{S}'=\wtil{M}\setminus\wtil{S}$,
the only difference between
\begin{equation*}
  P(z,\lambda) = \sum_{S\in\mc{P}(\mc{M}),1\in S} C_S
  \lambda^{\nu_{S'}}z^{k_S}, \qquad\quad
  \wtil{P}(z,\lambda) = \sum_{\wtil{S}\in\mc{P}(\wtil{\mc{M}})}
  \wtil{C}_{\wtil{S}}\lambda^{\tilde{\nu}_{\wtil{S}'}}z^{
    \tilde{k}_{\wtil{S}}},
\end{equation*}
is that each term in the former sum carries an extra factor
of $z^{k_1}$ when the sets $S$ and $\wtil{S}$ are matched up
in the natural way.  Thus $P(z,\lambda)=z^{k_1}\wtil{P}(z,\lambda)$.

Now consider the case $s_J=0$ with $J=\opn{argmax}_{2\le j\le
  M}\tau_j$.  This time $a_{J-1,J}=0$ due to $q_J=p_{J-1}$,
so $C_S=0$ unless $J-1$ and $J$ are both in $S$ or both in
$S'$.  Let us define
\begin{equation}
  \mc{P}_J(\mc{M}) = \{S\in\mc{P}(\mc{M})\;:\; J-1,J\in S
  \text{\; or\; }
  J-1,J\in S'\}.
\end{equation}
When we perform the sum over $S\in\mc{P}_J(\mc{M})$ to construct
$P(z,\lambda)$, we can consider $J-1$ and $J$ as a single unit.
The following factors always appear together in any $C_S$ that
contains one of them:
\begin{equation*}
  a_{J-1,m}a_{Jm} = \sqrt{\frac{(p_m-q_{J-1})(q_m - p_J)}{
    (q_m - q_{J-1})(p_m - p_J)}}, \qquad
  b_{J-1}b_J =
  \sqrt{\frac{q_{J-1}}{p_J}}e^{-ik_{J-1}x_{J-1,0}}e^{-ik_Jx_{J0}}.
\end{equation*}
Thus we can remove $p_{J-1}$ and $q_J$ from the sequence if
we short circuit the diagram
\begin{equation*}
  \xymatrix@C-2pc@R-1.5pc{
    0 \ar@/^1pc/[rr]^{s_1}
    & \le & q_1 \ar@/^1pc/[rr]^{k_1}
    & < & p_1 \ar@/^1pc/[rr]^{s_2}
    & \le & \cdots \ar@/^1pc/[rr]^{s_{J-1}}
    & \le & q_{J-1} \ar@/^1pc/[rr]^{k_{J-1}}
    & < & p_{J-1} \ar@/^1pc/[rr]^{s_J=0}
    & = & q_J \ar@/^1pc/[rr]^{k_J}
    & < & p_J \ar@/^1pc/[rr]^{s_{J+1}}
    & \le & q_{J+1} \ar@/^1pc/[rr]^{k_{J+1}}
    & < & \cdots \ar@/^1pc/[rr]^{s_M}
    & \le & q_M \ar@/^1pc/[rr]^{k_M}
    & < & p_M \\
    0 \ar@/_1pc/[rr]_{s_1}
    & \le & \tilde{q}_1 \ar@/_1pc/[rr]_{k_1}
    & < & \tilde{p}_1 \ar@/_1pc/[rr]_{s_2}
    & \le & \cdots \ar@/_1pc/[rr]_{s_{J-1}}
    & \le & \tilde{q}_{J-1} \ar@/_1pc/[rrrrrr]_{
      \tilde{k}_{J-1}=k_{J-1}+k_J}
    & & & < & & &
    \tilde{p}_{J-1} \ar@/_1pc/[rr]_{\tilde{s}_J=s_{J+1}}
    & \le & \tilde{q}_J \ar@/_1pc/[rr]_{\tilde{k}_J}
    & < & \cdots \ar@/_1pc/[rr]_{\tilde{s}_{M-1}}
    & \le & \tilde{q}_{M-1} \ar@/_1pc/[rr]_{\tilde{k}_{M-1}}
    & < & \tilde{p}_{M-1}
  }
\end{equation*}
and set
\begin{equation}
  \tilde{k}_{J-1}=k_{J-1}+k_J, \qquad
  \tilde{\nu}_{J-1}=\nu_{J-1}+\nu_J, \qquad
  \tilde{x}_{{J-1},0} = \frac{k_{J-1}x_{J-1,0} + k_Jx_{J0}}{k_{J-1}+k_J}.
\end{equation}
The other parameters are simply copied from the original sequence
as was indicated in (\ref{eqn:sJeq0}).  We then have
\begin{equation*}
  P(z,\lambda) = \sum_{S\in\mc{P}_J(\mc{M})} C_S
  \lambda^{\nu_{S'}}z^{k_S} =
  \sum_{\wtil{S}\in\mc{P}(\wtil{\mc{M}})}
  \wtil{C}_{\wtil{S}}\lambda^{\tilde{\nu}_{\wtil{S}'}}z^{
    \tilde{k}_{\wtil{S}}} = \wtil{P}(z,\lambda),
\end{equation*}
where again the sets $S$ and $\wtil{S}$ in these sums are in
natural 1-1 correspondence.

Finally, we verify that the parameters $\tilde{k}_j$, $\tilde{s}_j$
and $\tilde{\nu}_j$ of the reduced system are consistent
with the construction, i.e.~if $\wtil{M}\ge2$ and
we define $\tilde{J}=\opn{argmax}_{2\le
j\le\wtil{M}}\tilde{\tau}_j$, then
\begin{equation}
  \tilde{s}_j = \frac{\tilde{\tau}_{\tilde{J}}-\tilde{\tau}_j}{
    \tilde{\gamma_j}} + \frac{\tilde{\gamma}_{\tilde{J}}}{\tilde{\gamma}_j}
  \tilde{s}_{\tilde{J}}, \qquad
  (2\le j\le\wtil{M}, \;\;j\ne\tilde{J}).
\end{equation}
Recall that these equations are obtained by eliminating $\tilde{c}$
and $\tilde{\omega}$ from
\begin{equation} \label{eqn:til:kjc}
  \tilde{k}_j ( \tilde{c} + \tilde{c}_j ) = \tilde{\nu}_j\tilde{\omega},
  \qquad (1\le j\le \wtil{M}).
\end{equation}
In the first case where $s_1=0$, (\ref{eqn:s1eq0}) and (\ref{eqn:kj:c})
together with
\begin{equation}
  \tilde{c}_j = \tilde{p}_j + \tilde{q}_j = c_{j+1} - 2k_1, \qquad
  (1\le j\le\wtil{M})
\end{equation}
imply that (\ref{eqn:til:kjc}) is satisfied if we define
\begin{equation}
  \tilde{c} = c + 2k_1, \qquad \tilde{\omega}=\omega.
\end{equation}
Since we annihilate $k_1$ particles in this case, $\wtil{N}=N-k_1$,
and the mean
\begin{equation}
  \tilde{\alpha}_0 = \tilde{c}+2\wtil{N} = c+2N = \alpha_0
\end{equation}
remains unchanged in spite of the change in $c$ and $N$ (as it must
for the solution to vary continuously through the bifurcation).
In the remaining case where $s_J=0$, we have
\begin{equation}
  \tilde{c}_j = \begin{cases}
    c_j, & j<J-1, \\
    c_j+k_J=c_J-k_j,\;\; & j=J-1, \\
    c_{j+1}, & J\le j\le\wtil{M}.
  \end{cases}
\end{equation}
Thus, by (\ref{eqn:sJeq0}) and (\ref{eqn:kj:c}), all the equations in
(\ref{eqn:til:kjc}), except possibly $j=J-1$, are trivially satisfied if
we define
\begin{equation}
  \tilde{c}=c, \quad\qquad \tilde{\omega}=\omega.
\end{equation}
The remaining equation is
\begin{equation} \label{eqn:kjc2}
  (k_{J-1} + k_J)(c+c_{J-1}+k_J)=(\nu_{J-1} + \nu_J)\omega.
\end{equation}
To see that this is true, note that by (\ref{eqn:kj:c}),
\begin{equation}
  k_{J-1}(c+c_{J-1})=\nu_{J-1}\omega, \quad\qquad
  k_J(c+c_J)=\nu_J\omega.
\end{equation}
Adding these equations and using $k_Jc_J=k_J(c_{J-1}+k_{J-1}+k_J)$
gives (\ref{eqn:kjc2}), as required.  Since $\tilde{c}=c$
and $\wtil{N}=N$, the mean $\alpha_0=c+2N$ does not change as a result
of the bifurcation.

Thus, we have shown that when $s_1=0$ or $s_J=0$, the function
$u(x,t)$ in (\ref{eqn:u:a0:beta}) agrees with another function in the
hierarchy with $M$ reduced by one and the parameter $s_1$ or $s_J$
removed.  Continuing in this fashion, we can remove all the zero
indices, eventually yielding the case where the $p_j$ and $q_j$ are
distinct from one another, or the case that $M\in\{0,1\}$.  This shows
that $u(x,t)$ is a solution of (\ref{eqn:BO}), where we rely on
Theorem~\ref{thm1} to handle the reduction from $M=2$ to the traveling
wave case $M=1$, or the constant solution case $M=0$.
\end{proof}

\section{Examples} \label{sec:examples}
In this section we present several examples to illustrate the
types of bifurcation that occur in the hierarchy of time-periodic
solutions described in Theorem~\ref{thm2}.  We begin with the
simplest example that leads to a degenerate bifurcation, namely
\begin{equation}
  M=3, \qquad \vec{k}=(1,1,1), \qquad \vec{\nu}=(-2,-1,0).
\end{equation}
This solution and the three $M=2$ solutions connected to it have
parameters shown in Figure~\ref{fig:params}.  We hold the mean
$\alpha_0=0.544375$ fixed, which is the value used in several of the
numerical simulations in \cite{benj1,benj2}, and construct a single
bifurcation diagram showing all four solutions; see
Figure~\ref{fig:ripples}.  Note that paths B,C,D are actually 
parametrized by
\begin{equation}
  \tilde{s}_1^{(B)} = s_2, \quad
  \tilde{s}_2^{(B)} = s_3, \qquad
  \tilde{s}_1^{(C)} = s_1, \quad
  \tilde{s}_2^{(C)} = s_3, \qquad
  \tilde{s}_1^{(D)} = s_1, \quad
  \tilde{s}_2^{(D)} = s_2
\end{equation}
in formulas (\ref{eqn:mj:def})--(\ref{eqn:P:CS}), but we use the
original variables $s_j$ here for the convenience of making a single
bifurcation diagram.  The one confusing aspect of doing this is that
we obtain different traveling waves depending on the order in which
we set the $s_j$ to zero.  This is why we drew two axes for $s_2$ and
$s_3$ in Figure~\ref{fig:ripples}.  For example, if we start on path A
and decrease $s_1$ to zero (moving to path B) and then decrease $s_3$
to zero, we obtain the traveling wave bifurcation $(2,-1,1,1)$;
however, if we first set $s_3=0$ (moving to path~D) and then set
$s_1=0$, we obtain the bifurcation $(2,-1,2,3)$.  Both traveling waves
have $N=2$ humps and speed index $\nu=-1$, but the amplitude and
period of the two solutions are different as they have different
bifurcation indices.  Similarly, although starting on path A and
setting $s_2$ and $s_3$ to zero in either order leads to the same
stationary solution, the period $T$ is different depending on whether
we follow path B to $(1,0,1,1)$ or path C to $(1,0,2,3)$.  This only
happens at the bottom ($\wtil{M}=1$) level of the hierarchy, where
$\tilde{k}_1$, $\tilde{\nu}_1$ and $\tilde{s}_1$ are not sufficient to
uniquely determine the traveling wave; if we start with $M\ge4$ and
follow two paths down several levels to $\wtil{M}\ge2$ with the same
parameters $\tilde{k}_j$, $\tilde{\nu}_j$ and $\tilde{s}_j$, the
resulting solution is independent of the path.

\begin{figure}[t]
\fbox{\parbox{5.9in}{
\begin{alignat*}{8}
  &&\underline{\text{path A}}\;\;&& 
  &&\underline{\text{path B}}\;\;&& 
  &&\underline{\text{path C}}\;\;&& 
  &&\underline{\text{path D}}\;\; \\
  k&=& 1, 1,1,& &   k&=& 1,1,& &   k&=& 2,1,& &   k&=& 1, 2, \\
  \nu&=&\,-2,-1,0,& & \nu&=&-1,0,& & \nu&=&-3,0,& &  \nu&=&\,-2,-1, \\
    m&=&    1,1,& &   m&=&   1,& &   m&=&   3,& &   m&=&    3, \\
 \tau&=&    2,2,& &\tau&=&   2,& &\tau&=&   2,& &\tau&=&    2, \\
 \gamma&=&2,2,&&\gamma&=&2,&&\gamma&=&4/3,&&\gamma&=&4/3,\\
 s_3&=& s_2,& &s_1&=& 0,& &s_2&=&0,& &s_3&=&0, \\
 \alpha_0&=&1-2s_1-&4s_2,\quad&
 \alpha_0&=&\,1-2s_2&-2s_3,\quad&
 \alpha_0&=&\,1-2s_1&-2s_3,\quad&
 \alpha_0&=&1-2s_1&-\frac{8}{3}s_2, \\
 \omega&=&2+2s_2,&&
 \omega&=&2+2s_3,&&
 \omega&=&\,2+\frac{4}{3}s_3,&&
 \omega&=&\,2+\frac{4}{3}s_2.&
\end{alignat*}}}
\caption{Parameters of four paths of time-periodic solutions connected
  by bifurcations.}\label{fig:params}
\end{figure}
\begin{figure}[p]
\begin{center}
\mypsdraft
\includegraphics[width=.44\linewidth,trim=0 0 0 0,clip]{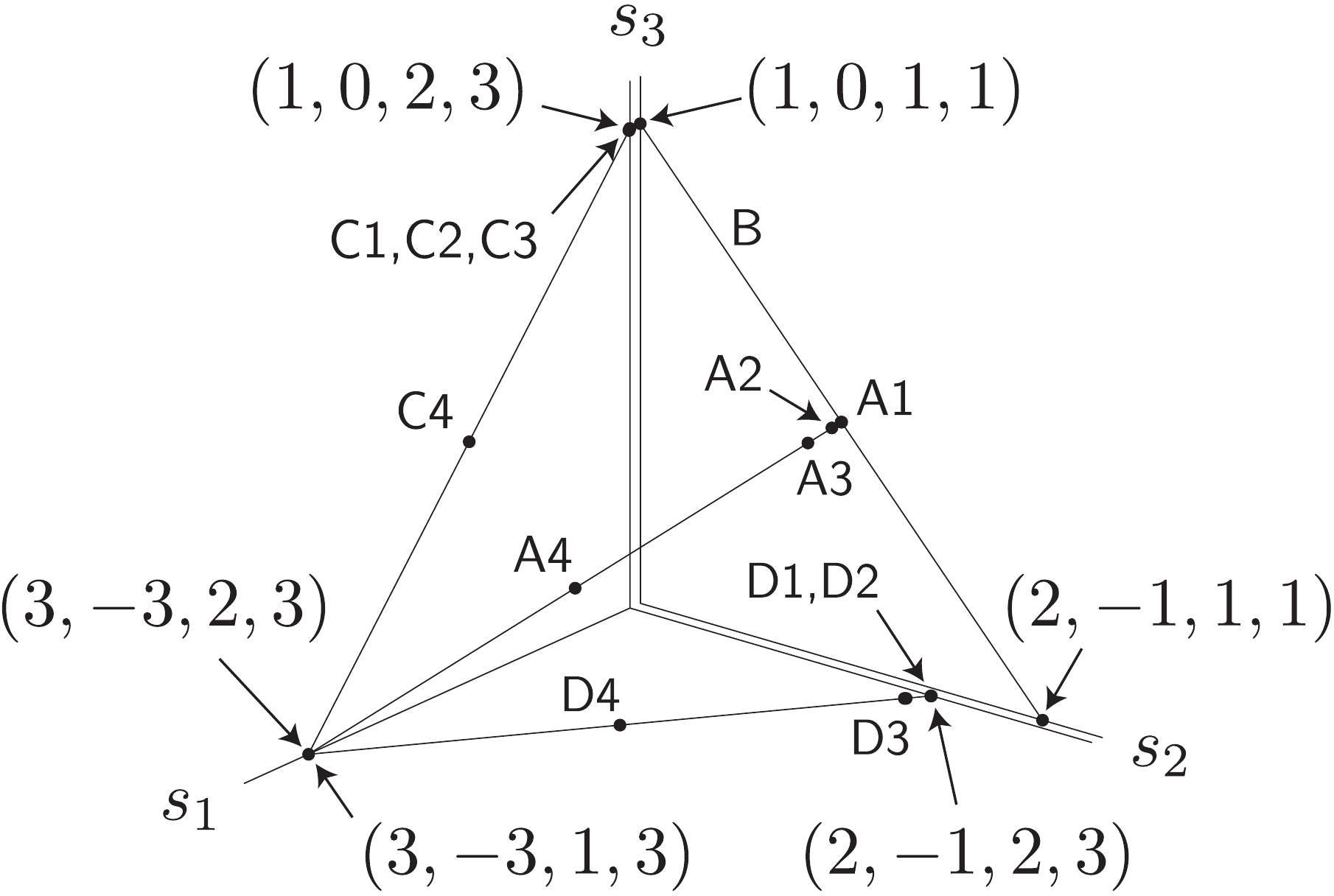}\qquad
\includegraphics[width=.5\linewidth,trim=0 0 0 0,clip]{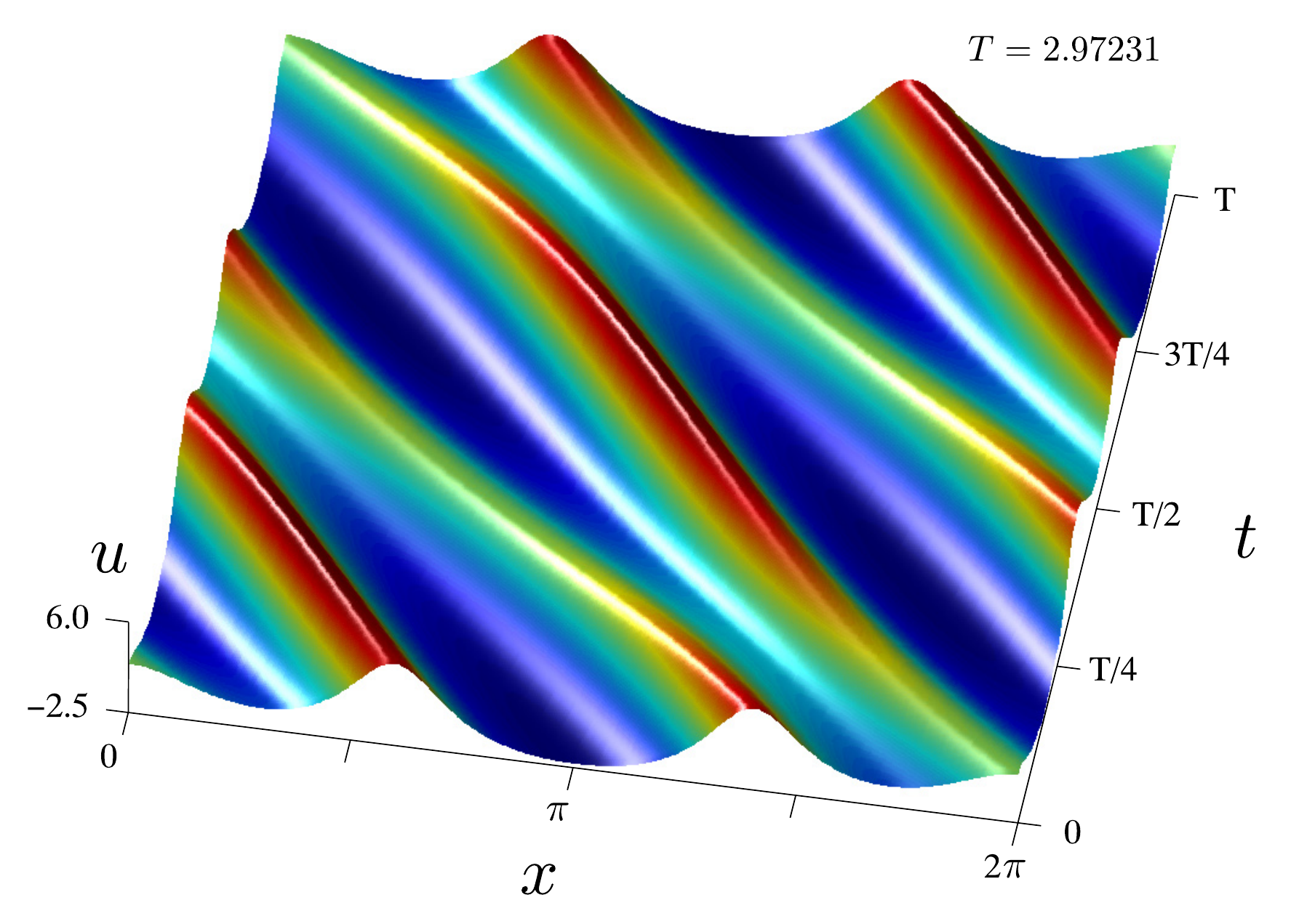}
\mypsfull
\end{center}
\caption{\emph{Left:} degenerate bifurcation from $M=1$ to
$M=2$ (paths C and D) and $M=3$ (path A) with $\alpha_0$ held fixed.
\emph{Right:} three dimensional plot of the solution labeled D4.
}
\label{fig:ripples}
\end{figure}
\begin{figure}[p]
\begin{center}
\mypsdraft
\includegraphics[width=\linewidth,trim=0 0 0 0]{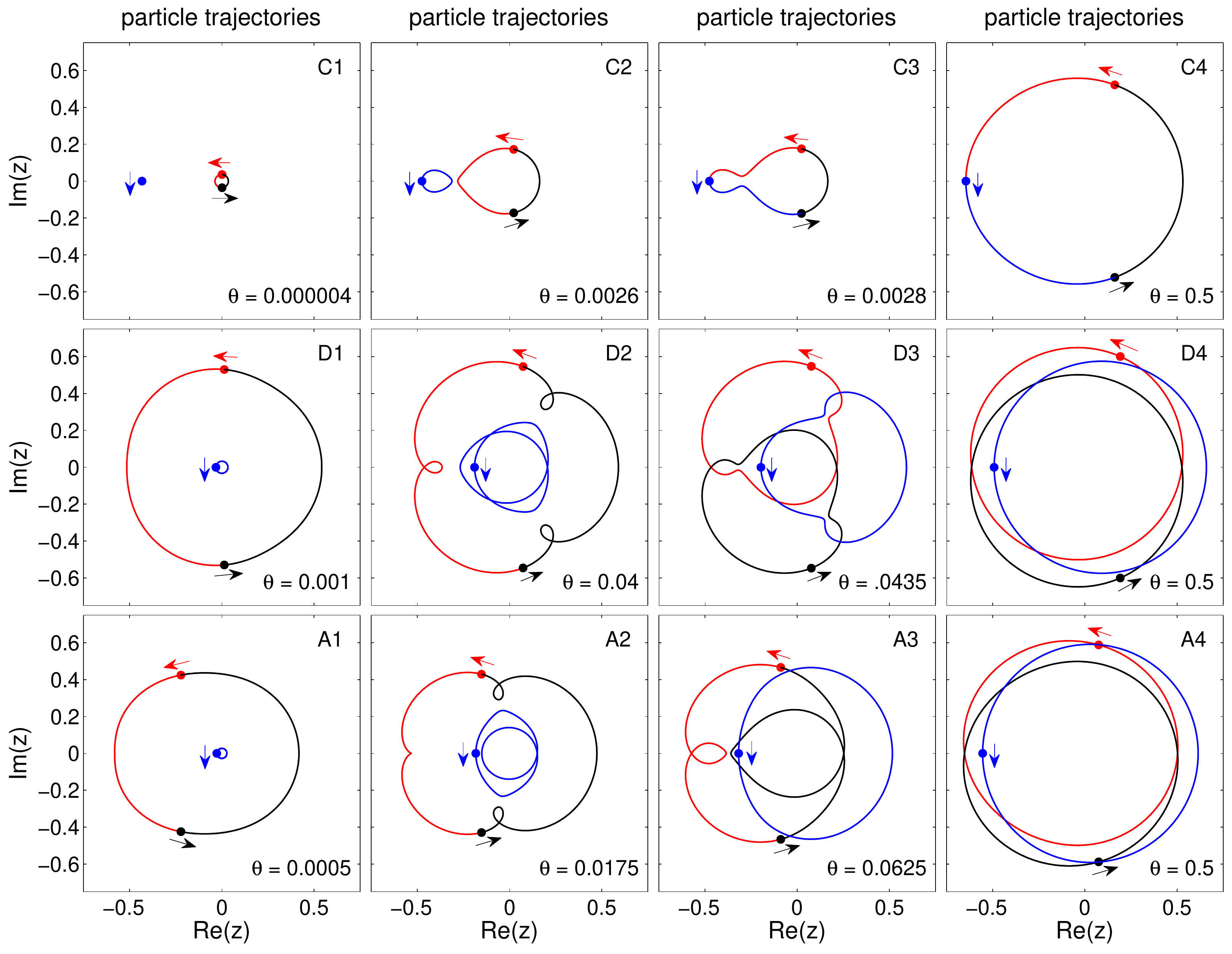}
\mypsfull
\end{center}
\caption{Particle trajectories along paths C,D,A in the bifurcation
diagram of Figure~\ref{fig:ripples}.
}\label{fig:beta3}
\end{figure}

In Figure~\ref{fig:beta3}, we plot the particle trajectories of
several solutions on paths C, D and A in the bifurcation diagram of
Figure~\ref{fig:ripples}.  We parametrize each path linearly by a
variable $\theta\in[0,1]$.  For example, on path A,
\begin{equation}
  s_1 = \frac{1-\alpha_0}{2}\theta, \qquad
  s_2 = s_3 = \frac{1-\alpha_0}{4}(1-\theta), \qquad 0\le\theta\le 1.
\end{equation}
Path C connects the one-hump stationary solution to the three-hump
traveling wave.  When $\theta=4.0\times 10^{-6}$, two particles have
nucleated at the origin and execute small, nearly circular orbits
around each other while the third particle travels around its original
resting position.  As $\theta$ increases, the orbits deform and
coalesce into a single path, as shown in the middle two panels of this
row.  At the critical value $\theta=0.002649485$, the particles
collide at $t=T/6$, $t=3T/6$ and $t=5T/6$, so the solution of the ODE
(\ref{eqn:beta:ode}) ceases to exist for all time;
nevertheless, $u(x,t)$ in (\ref{eqn:u:a0:beta}) remains smooth and
satisfies (\ref{eqn:BO}) for all $t$.  As $\theta$ increases to 1, the
common trajectory of the three particles becomes more and more
circular until the traveling wave is reached, where it is exactly
circular.  The solutions on this path are reducible in the sense that
their natural period is $1/3$ of the period $T$ used here.  However,
in order to bifurcate to paths A and D, we have to use this solution
rather than the reduced solution.

Path D connects the two-hump traveling wave with speed index $\nu=-1$
to the three-hump traveling wave with speed index $\nu=-3$.  When we
bifurcate from the two-hump traveling wave, a new particle nucleates
at the origin and grows in amplitude until its trajectory joins up
with the orbits of the outer particles.  As $\theta$ increases
further, the three orbits become nearly circular and eventually
coalesce into a single circular orbit at the three-hump traveling
wave.  Note that the particles on path D follow different
trajectories, which leads to a ``braided'' effect in the peaks and
troughs of the solution $u(x,t)$ shown in Figure~\ref{fig:ripples};
unlike path C, these solutions are not reducible to a shorter period.

Path A connects an interior bifurcation on path B to this same
three-hump traveling wave.  When $\theta=0.0005$, a particle has
nucleated at the origin without destroying the periodicity of the
orbit of the other two particles.  This path involves two topological
changes in the particle trajectories, as shown in the middle two
panels of the bottom row of Figure~\ref{fig:beta3}.  
As
$\theta\rightarrow1$, this solution also approaches the three-hump
traveling wave.  This is interesting because, up to a phase shift in
space and time, the linearized Benjamin-Ono equation about this
traveling wave \cite{benj2} has only two linearly independent,
time-periodic solutions corresponding to the bifurcations $(3,-3,1,3)$
and $(3,-3,2,3)$; this degenerate bifurcation is not predicted by
linear theory.

\begin{figure}[p]
\begin{center}
\mypsdraft
\includegraphics[width=.65\linewidth,trim=0 0 0 0,clip]{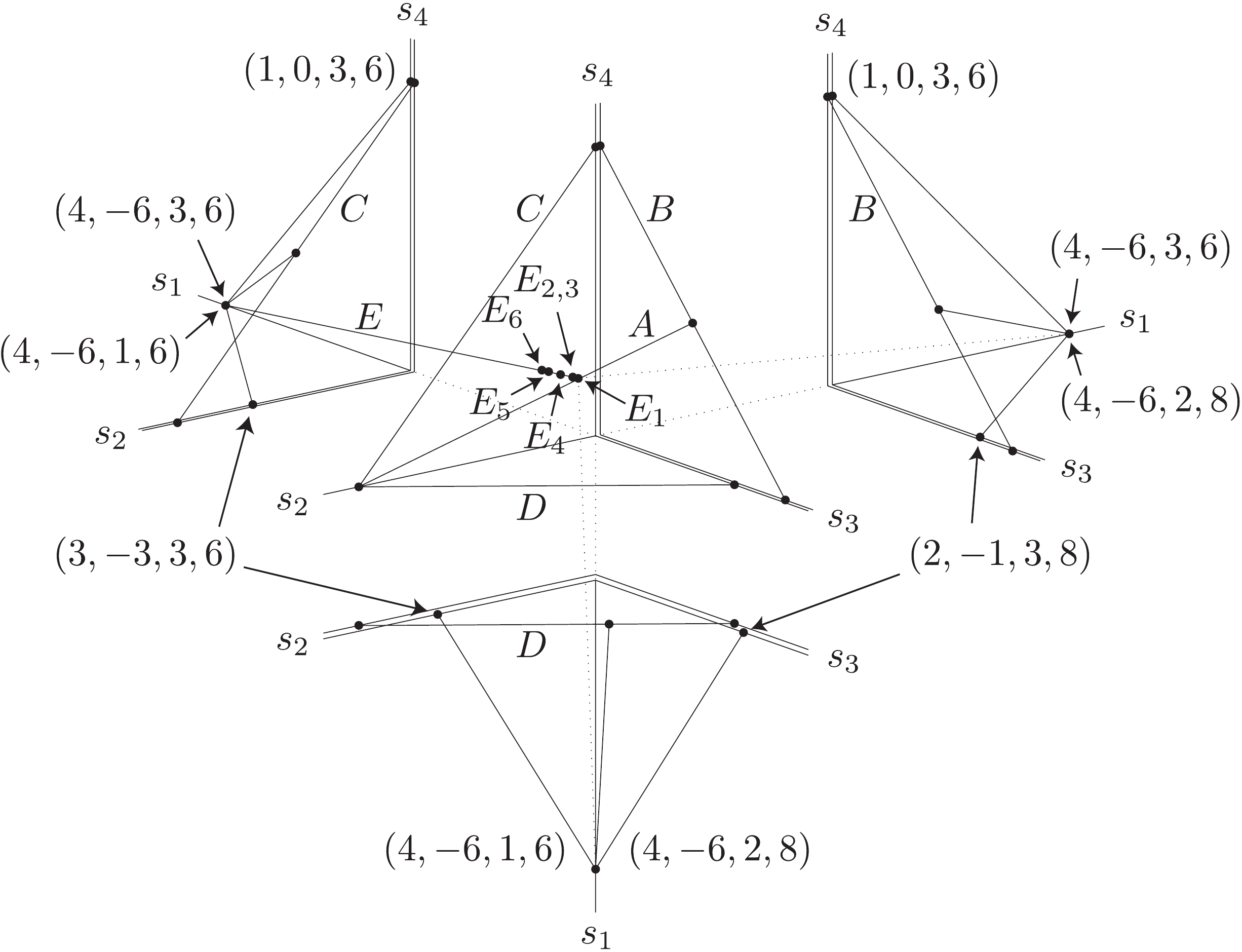}
\mypsfull
\end{center}
\caption{Degenerate bifurcation from $M=1$ to
$M=2,3,4$ with $\alpha_0$ held fixed.
}
\label{fig:bifurE}
\end{figure}
\begin{figure}[p]
\begin{center}
\mypsdraft
\includegraphics[width=\linewidth,trim=0 0 0 0]{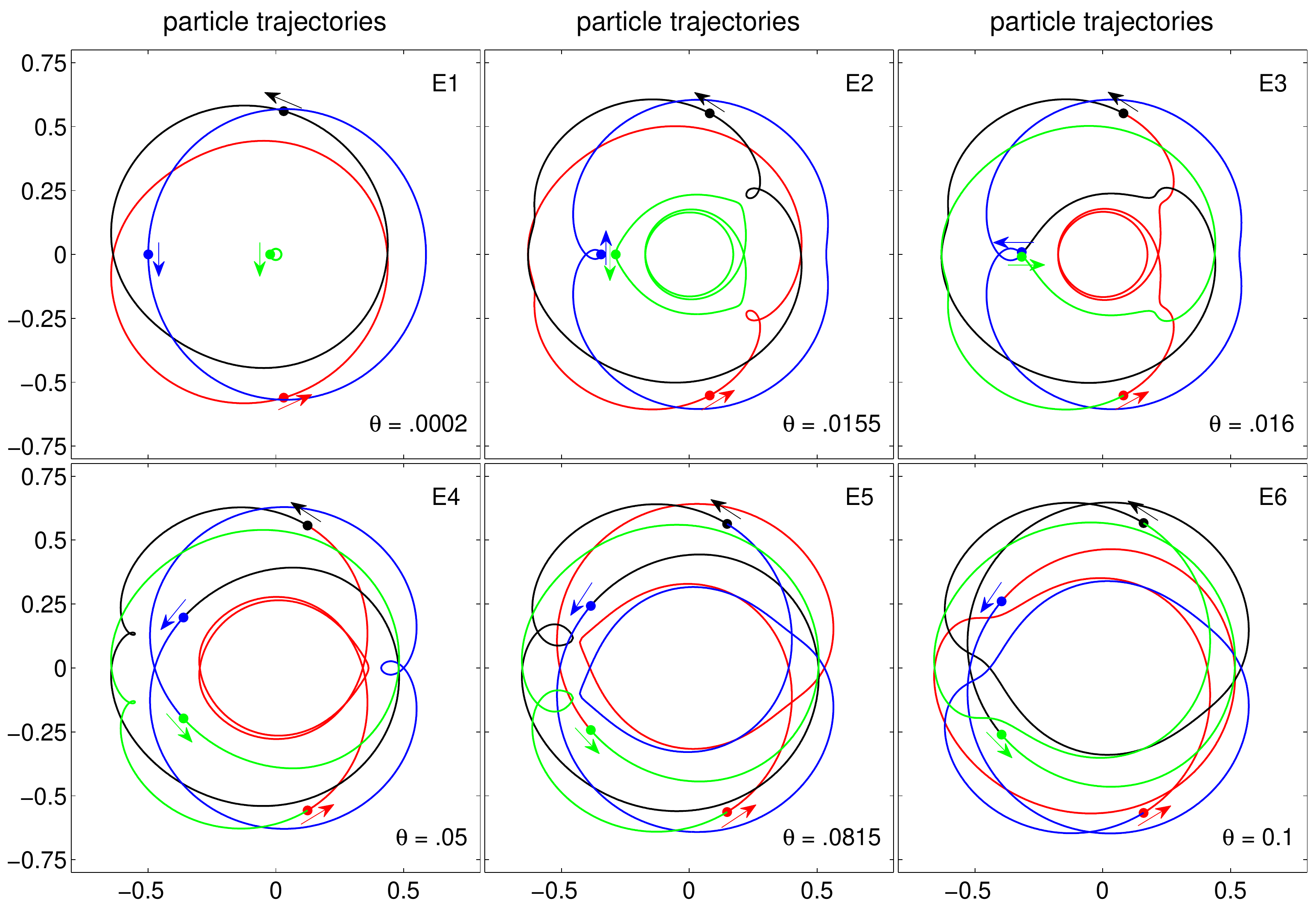}
\mypsfull
\end{center}
\caption{Particle trajectories along path E in the bifurcation
diagram of Figure~\ref{fig:bifurE}.
}\label{fig:particleE}
\end{figure}

In a similar way, we can construct a bifurcation from a traveling
wave to an arbitrary level of the hierarchy by taking
\begin{equation}
  M \text{ arbitrary}, \qquad \vec k=(1,1,\dots,1), \qquad
  \nu=(-M,-M+1,\dots,-2,-1,0).
\end{equation}
We find that $m_j=1$, $\tau_j=2$ and $\gamma_j=2$ for $2\le j\le M$,
so any subset of the indices $J=2,\dots,M$ can be removed to obtain a
solution at a lower level of the hierarchy.  If all the indices are
removed, we obtain a traveling wave with $2^{M-1}-1$ bifurcations to
higher levels of the hierarchy, but only $M-1$ of them (to the second
level) are predicted by linear theory.  The case $M=4$ is depicted in
the bifurcation diagram of Figure~\ref{fig:bifurE}, where each
tetrahedron contains paths with one of the $s_j$ set to zero, and the
outermost point of the outer three tetrahedra corresponds to one and
the same traveling wave.  Linear theory predicts the bifurcations
$(4,-6,1,6)$, $(4,-6,2,8)$, $(4,-6,3,6)$ from this traveling wave to
the $M=2$ level of the hierarchy, but does not predict the three $M=3$
families of solutions that connect this traveling wave to interior
bifurcations on paths B, C and D, nor the $M=4$ solution connecting
this traveling wave to path A from the previous example.  In
Figure~\ref{fig:particleE}, we show six solutions on this $M=4$ path
labeled E in the bifurcation diagram.  When $\theta=0^+$, (where path
E meets path A), a fourth particle nucleates at the origin without
destroying periodicity of the other three.  As $\theta$ increases to 1,
the trajectories of the particles $\beta_j(t)$ undergo several
topological changes (that determine which particles exchange positions
over one period) until the trajectories coalesce into a single
circular orbit at the degenerate traveling wave (with each particle
moving counter-clockwise one and a half times per period).

Finally, in Figure~\ref{fig:bifurG}, we show a path of solutions
at level $M=3$ that connects two $M=2$ solutions by interior
(non-degenerate) bifurcations.  The parameters of this path
are
\begin{equation}
  M=3, \qquad \vec k=(1,1,1), \qquad \vec \nu = (-2,0,1).
\end{equation}
We parametrize this path by
\begin{equation} \label{eqn:s1s3:G}
  s_1 = \frac{\theta}{2}, \qquad s_2 = \frac{3-\theta}{2}, \qquad
  s_3 = \frac{1-\theta}{4}, \quad\qquad (0\le \theta\le 1).
\end{equation}
When $\theta=0$, we obtain the $M=2$ family of solutions that
bifurcates from a two-hump, right traveling wave with indices
$(2,1,1,1)$.  For this family of solutions, the mean is related to the
parameters $\tilde{s}_1=s_2$ and $\tilde{s}_2=s_3$ via
$\alpha_0=3-2s_2$; hence, holding the mean fixed requires that $s_2$
remains constant.  We will not reach the one-hump, right traveling wave
at the bifurcation $(1,1,1,1)$ unless we increase $\alpha_0$ to 3.  As
we increase $\theta$ in (\ref{eqn:s1s3:G}), the trajectory of the
particle that nucleates at the origin at $\theta=0^+$ grows and merges
with the trajectories of the original two particles through three
topological changes: one at F1, one not shown, and one at F3.
Eventually, when $\theta=1$, path F joins path G connecting
$(2,1,2,5)$ to $(3,-1,1,5)$.  Since the two-hump traveling wave moves
to the right while the three-hump traveling wave moves to the left,
the solution at F3 involves particles moving clockwise for part of
their orbit and counter-clockwise at other times, leading to an
interesting three-particle trajectory with 5-fold symmetry.

\begin{figure}[t]
\begin{center}
\mypsdraft
\includegraphics[width=.35\linewidth,trim=0 0 0 0,clip]{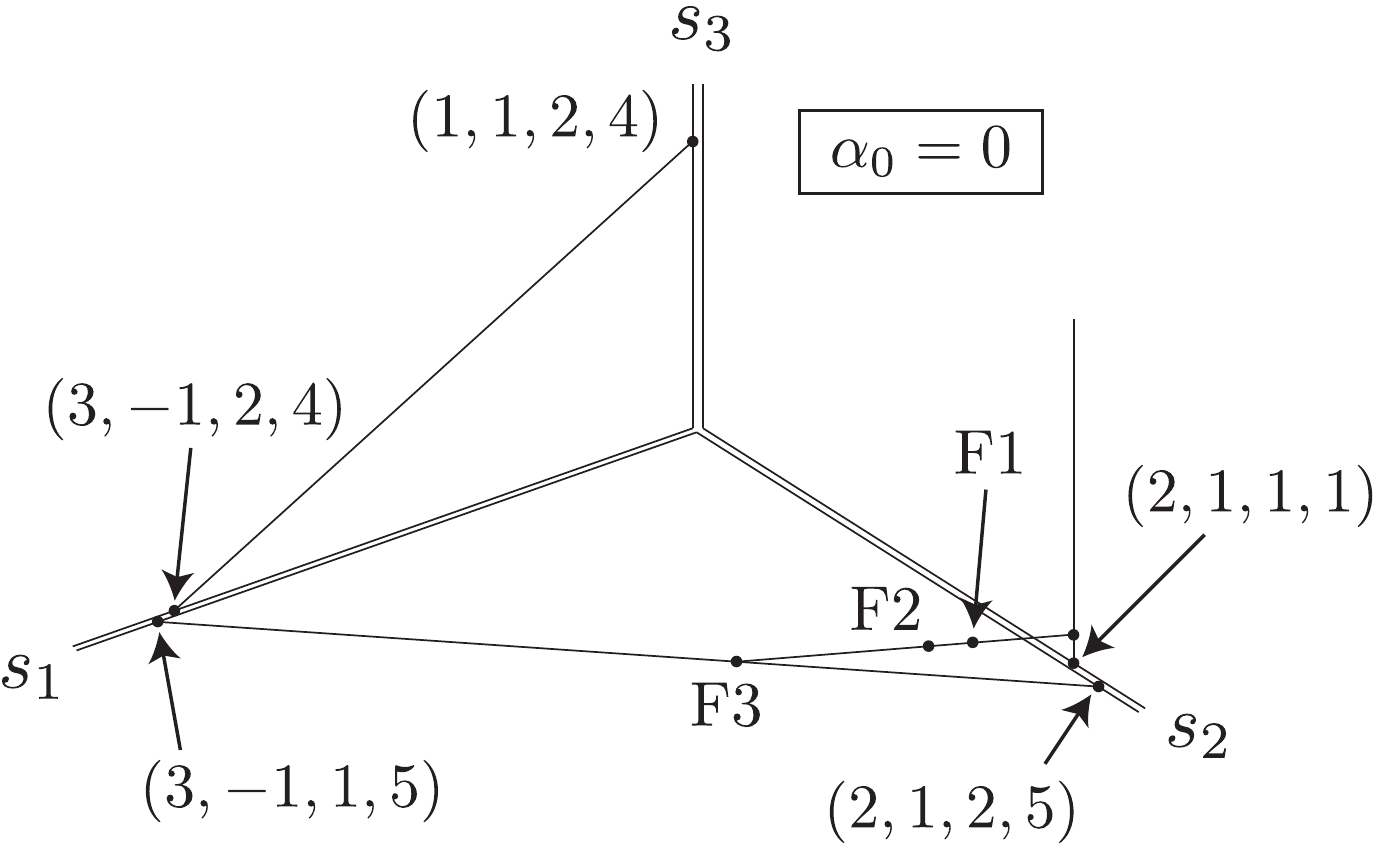}\quad
\includegraphics[width=.6\linewidth,trim=0 0 0 0,clip]{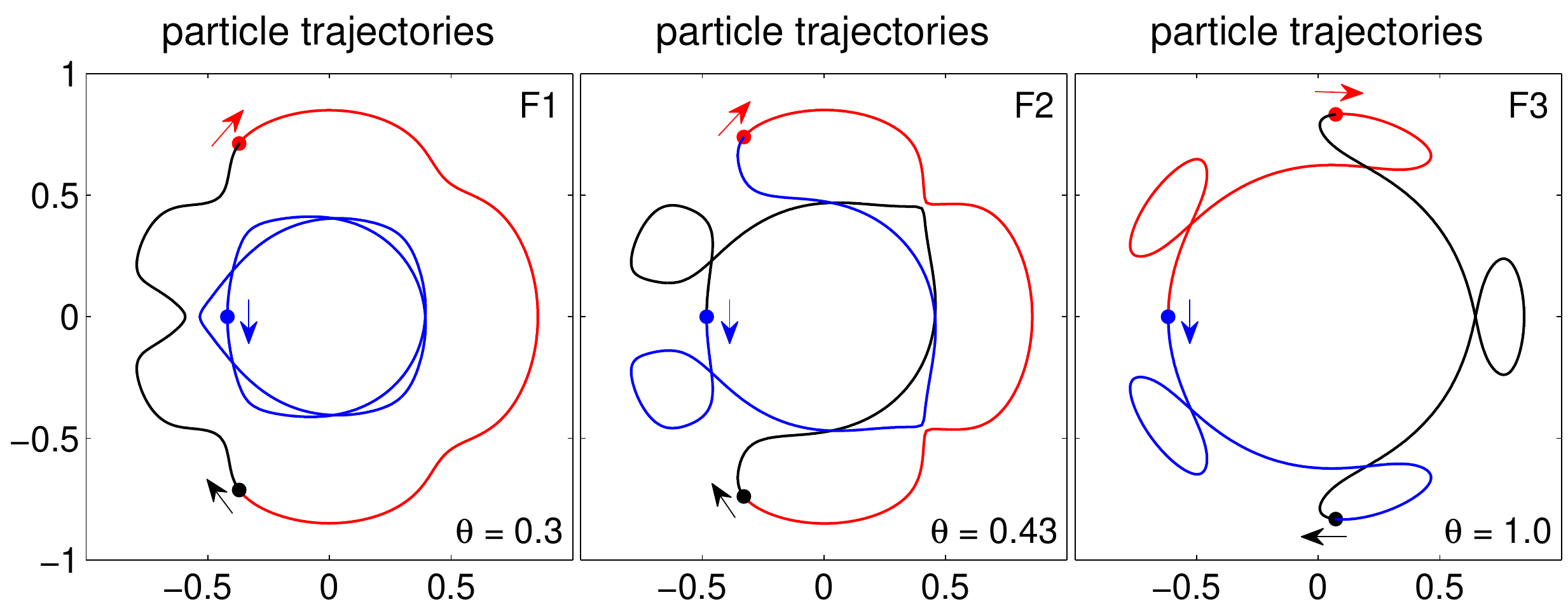}
\mypsfull
\end{center}
\caption{\emph{Left:} Bifurcation diagram showing a path of time-periodic
solutions at level $M=3$ connecting two $M=2$ solutions. \emph{Right:}
Three solutions on this path.
}
\label{fig:bifurG}
\end{figure}

\appendix

\section{Zeros of the polynomial $P(\cdot,\lambda)$}
\label{appendix:roots}

As mentioned in the proof of Theorem~\ref{thm2}, there is a small gap
in the paper \cite{satsuma:ishimori:79} showing that the multiperiodic
solutions (\ref{eqn:u:satsu}) satisfy the Benjamin-Ono equation, for
the bilinear formalism used to derive these solutions requires that
the zeros of $f$ (or $f'$) in (\ref{eqn:f:satsu1}) and
(\ref{eqn:f:satsu2}) lie in the upper (or lower) half of the complex
plane.  In this appendix, we prove the equivalent assertion (in the
case that all the $k_i$ are integers) that the zeros of the polynomial
$P(\cdot,\lambda)$ in (\ref{eqn:P:CS}) lie inside the unit disk.  The
two key ideas of this proof, namely showing that $P$ (or $f$) has a
representation as a determinant, and that the matrix is non-singular
for $|z|\ge1$,
are essentially due to Matsuno \cite{matsuno:04} and
Dobrokhotov/Krichever \cite{dobro:91}, respectively.

\begin{theorem}\label{thm3}
Suppose $M\ge2$,\, $k_1,\dots,k_M\in\mathbb{N}$,\,
$\nu_1,\dots,\nu_M\in\mathbb{R}$,\, $x_{10},\dots,x_{M0}\in\mathbb{R}$,
and
\begin{equation}
  0 < q_1 < p_1 < q_2 < p_2 < \cdots < q_M < p_M.
\end{equation}
Let $\mc{M}=\{1,\dots,M\}$ and define the polynomial
\begin{equation}
  P(z,\lambda) = \sum_{S\in\mc{P}(\mc{M})} C_S \lambda^{\nu_{S'}}z^{k_S},
  \qquad
  C_S = \Biggl( \prod_{(i,j) \in S\times S'} a_{ij}\Biggr)
  \Biggl( \prod_{j\in S'} b_j \Biggr),
\end{equation}
where $S'=\mc{M}\setminus S$, $k_S = \sum_{i\in S}k_i$,
$\nu_{S'} = \sum_{j\in S'} \nu_j$, and
\begin{equation} \label{eqn:ab:def2}
  a_{ij} = \sqrt{\frac{(p_j-q_i)(q_j-p_i)}{(q_j-q_i)(p_j-p_i)}}, \qquad
  b_j = \sqrt{\frac{q_j}{p_j}}e^{-ik_jx_{j0}}.
\end{equation}
Then all the zeros $\beta_l$ of $P(\cdot,\lambda)$ lie inside the
unit disk $\Delta\in\mathbb{C}$ provided $|\lambda|=1$.
\end{theorem}

\begin{proof}
First we show that
\begin{equation} \label{eqn:P:det}
  P(z,\lambda) = \bigg[ \prod_{j=1}^M (p_j-q_j)\biggr]
  \bigg[\prod_{j=1}^M b_j\lambda^{\nu_j}\biggr]
  \bigg[\prod_{i<j}^{(\mc{M})} a_{ij}^2\biggr] \det R(z,\lambda),
\end{equation}
where the $M\times M$ matrix $R(z,\lambda)$ has entries
\begin{equation}
  R_{ij}(z,\lambda) = r_i(z,\lambda) \delta_{ij} + \frac{1}{p_i-q_j}, \qquad
  r_i(z,\lambda) =
  \frac{b_i^{-1}\lambda^{-\nu_i}z^{k_i}}{p_i-q_i}\prod_{j\ne i}^M
  a_{ij}^{-1}.
\end{equation}
The symbol $\prod_{j\ne i}^M$ indicates a product over $j\in\mc{M}$
omitting $j=i$, while $\prod_{i<j}^{(\mc{M})}$ is a product over all pairs
$(i,j)\in \mc{M}^2$ such that $i<j$.  By expanding $\det R=\sum_\sigma
\opn{sgn}(\sigma)R_{i,\sigma(i)}$ and collecting like products of the
$r_i(z,\lambda)$, we find that
\begin{equation}
  \det R = \sum_{S\in\mc{P(\mc{M})}} \bigg[\prod_{i\in S} r_i\biggr]
  \det R_{S'}, \qquad \big(R_{S'}\big)_{ij} = \frac{1}{p_{S'_i} - q_{S'_j}}.
\end{equation}
Here $\{S'_1,\dots,S'_n\}$ is an enumeration of $S'$, i.e.~$R_{S'}$ is
the $n\times n$ (Cauchy) matrix obtained by removing the rows and
columns with indices in $S$ from the Cauchy matrix
$\big\{(p_i-q_j)^{-1}\}_{i,j=1}^M$, and $\det R_\varnothing$ is taken
to be 1.  The determinant of a Cauchy matrix is well-known, giving
\begin{equation}
  \det R_{S'} = \frac{
    \prod_{i<j}^{(S')}
    (p_i - p_j)(q_j - q_i)}{
    \prod_{i,j\in S'}(p_i-q_j)} =
  \bigg[\prod_{j\in S'}\frac{1}{p_j - q_j}\biggr]
  \prod_{i<j}^{(S')} a_{ij}^{-2}.
\end{equation}
Thus, the right hand side of (\ref{eqn:P:det}) is equal to
\begin{equation}
  \sum_{S\in\mc{P}(\mc{M})}
  \bigg[ \prod_{j\in S'} b_j \lambda^{\nu_j}\biggr]
  \bigg[ \prod_{i\in S} z^{k_i}\biggr]
  \bigg[ \prod_{i\in S}\prod_{j\ne i}^M a_{ij}^{-1} \biggr]
  \bigg[ \prod_{i<j}^{(\mc{S'})} a_{ij}^{-2} \biggr]
  \bigg[ \prod_{i<j}^{(\mc{M})} a_{ij}^2 \biggr].
\end{equation}
If $i$ and $j$ are both in $S$ or both in $S'$, the terms $a_{ij}^2$
in the final product cancel with corresponding terms in one of the
previous two products.  If $i\in S$ and $j\in S'$, one of the factors
of $a_{ij}$ (or $a_{ji}$ if $i>j$) in the final product cancels with
$a_{ij}^{-1}$ in the middle product, leaving behind
$\prod_{(i,j)\in S\times S'} a_{ij}$, as required.

Next we show that $R(z,\lambda)$ is invertible for $|z|\ge1$ and
$|\lambda|=1$.  Fix such a $z$ and $\lambda$.  Define $d_i=
|b_i^{-1}\lambda^{-\nu_i}z^{k_i}|$ so that
\begin{equation}
  |r_i|^2 = \frac{d_i^2}{(p_i - q_i)^2}\prod_{j\ne i}^M
  \frac{(q_j-q_i)(p_j-p_i)}{(p_j - q_i)(q_j - p_i)}, \qquad
  d_i > 1.
\end{equation}
Suppose for the sake of contradiction that there is a non-zero
vector $\gamma\in\mathbb{C}^M$ such that $R\gamma=0$.  This means
that
\begin{equation*}
  \sum_{j=1}^M\left(r_i\delta_{ij} + \frac{1}{p_i-q_j}\right)\gamma_j =
  r_i\gamma_i + \psi(p_i) = 0, \quad (1\le i\le M), \qquad
  \psi(k) := \jd\sum_{j=1}^M \frac{\gamma_j}{k - q_j}.
\end{equation*}
Then we define
\begin{equation}
  \phi(k) = \left( \sum_i \frac{\gamma_i}{k-q_i} \right)
  \left( \sum_m \frac{\bar{\gamma}_m}{k-q_m} \right)
  \left( \prod_j \frac{k - q_j}{k - p_j} \right)
\end{equation}
and observe that
\begin{equation}
  \opn{res}_{k=p_i}\phi +
  \opn{res}_{k=q_i}\phi =
  \frac{|\gamma_i|^2}{p_i - q_i}\bigg[ \prod_{j\ne i}
    \frac{q_i - q_j}{q_i - p_j}\biggr]\big( d_i^2 - 1 ) \ge 0,
  \qquad (1\le i\le M)
\end{equation}
where the inequality is strict if $\gamma_i\ne0$ and
we used $|\psi(p_i)|^2 = |r_i|^2|\gamma_i|^2$.  Thus, the
sum of all the residues of $\phi(k)$ is strictly positive,
contradicting $\phi(k)=O(k^{-2})$ as $k\rightarrow\infty$.
\end{proof}

\section{Uniqueness of Periodic Traveling Waves}
\label{appendix:unique}

In this section we elaborate on the paper \cite{amick:toland:BO} showing
that the only traveling wave solutions of the Benjamin-Ono equation
are the solitary and periodic wave solutions found by Benjamin in
\cite{benjamin:67}.  The purpose of this section is to modify their
argument to prove that we have found all $2\pi$-periodic traveling
solutions, and to simplify part of their analysis.

Consider any non-constant, $2\pi$-periodic traveling solution of
(\ref{eqn:BO}).  After a transformation of the form
(\ref{eqn:add:const}), we may assume $u(x)$ is a stationary solution
satisfying
\begin{equation}\label{eqn:BO:stat}
  uu_x = Hu_{xx}, \qquad\quad  (x\in\mathbb{R}\big/2\pi\mathbb{Z}).
\end{equation}
After a translation, we may assume $u_x(0)=0$.  Integrating once, there
is a constant $c$ such that
\begin{equation}\label{eqn:BO:stat2}
  \frac{1}{2}u^2 = Hu_x + \frac{c^2}{2}, \qquad (c>0).
\end{equation}
The integration constant must be positive since $Hu_x$ is the
derivative of a periodic function while the left hand side is
positive.  Now define the holomorphic function
\begin{equation}
  f_1(z) = \frac{1}{\pi}\int_0^{2\pi} \frac{u(\theta)e^{i\theta}}
  {e^{i\theta}-z}\,d\theta, \qquad (z\in\Delta).
\end{equation}
A direct calculation using Fourier series shows that
\begin{equation}
  f_1(z^+) = \alpha_0 + u(\theta) + iHu(\theta), \qquad
  z=e^{i\theta}, \qquad \alpha_0 = \frac{1}{2\pi}\int_0^{2\pi}
  u(\theta)\,d\theta,
\end{equation}
where $f_1(z^+)$ is the limit of $f_1(\zeta)$ as $\zeta$ approaches
$z$ from the inside.  Now define
\begin{equation}
  f_2(z) = f_1(e^{iz}) - \alpha_0, \qquad \imag z \ge 0.
\end{equation}
Then $f_2(z)$ is analytic and bounded in the upper half-plane and
satisfies
\begin{equation}
  f_2(x) = u(x) + iHu(x), \qquad x\in\mathbb{R}.
\end{equation}
Next we extend $u(x)$ to the upper half-plane via
$u(x,y)=\real\{f_2(x+iy)\}$ and define
\begin{equation}
  U(x,y) = \frac{u(x/c,y/c)+c}{2c},
\end{equation}
where $c$ was determined by $u(x)$ in (\ref{eqn:BO:stat2}).  We then
have
\begin{equation*}
  u_y(x,0) = \real\{if_2'(x)\}=-Hu_x(x)=\frac{c^2-u(x)^2}{2}, \qquad
  U_y(x,0) = U(x,0) - U(x,0)^2.
\end{equation*}
Amick and Toland \cite{amick:toland:BO} showed that any non-constant,
bounded, harmonic function $U(x,y)$ defined in the upper half-plane and
satisfying the nonlinear Neumann boundary condition $U_y=U-U^2$
on the real axis as well as $U_x(0,0)=0$, is given by
\begin{equation}
  U(x,y) = \real\{f(x+iy)\}, \qquad (y\ge0),
\end{equation}
where $f(z)$ is the (unique) solution of the complex ordinary differential
equation
\begin{equation}\label{eqn:f:trav:B}
  \frac{df}{dz}(z) = \frac{i}{2}\big[f(z)^2 - a^2\big], \qquad
  f(0)=U(0,0), \qquad a = \sqrt{2U(0,0) - U(0,0)^2}
\end{equation}
over the upper half-plane.
Moreover, such a solution $U(x,y)$ will satisfy $U(0,0)\in(0,1)\cup(1,2]$,
and the case $U(0,0)=2$ corresponds to the solitary wave solution
$U(x,0)=2/(1+x^2)$, which is ruled out by the assumption that $u$ in
(\ref{eqn:BO:stat}) is periodic.  Rather than treat the cases
$U(0,0)\in(0,1)$ and $U(0,0)\in(1,2)$ separately as was done in
\cite{amick:toland:BO}, we choose the unique $\beta\in(-1,1)$ such
that
\begin{equation}
  U(0,0) = \frac{(1+\beta)^2}{1+\beta^2}, \qquad\quad
  a = \frac{1-\beta^2}{1+\beta^2}
\end{equation}
and check directly that
\begin{equation}
  f(z) = a\left[1 + \frac{2\beta}{e^{-iaz} - \beta}\right]
\end{equation}
satisfies (\ref{eqn:f:trav:B}).  Since $U(x,y) = \real\{f(x+iy)\}$
is $x$-periodic with (smallest) period $\frac{2\pi}{a}$ while
$u(x,y)=[2cU(cx,cy)-c]$ is $x$-periodic with period $2\pi$,
it must be the case that $c=N/a$ for some positive integer $N$.
But then
\begin{equation}
  u(x)=\real\{Nf_3(Nx)\}, \qquad
  f_3(z) = \frac{2}{a}f\left(\frac{z}{a}\right)-\frac{1}{a} =
  \frac{1-3\beta^2}{1-\beta^2} + \frac{4\beta e^{iz}}{1-\beta e^{iz}},
\end{equation}
i.e.~$u(x)$ is one of the $N$-hump stationary solutions discussed
in Section~\ref{sec:trav}.

\bibliographystyle{plain}
\bibliography{refs}

\end{document}